\documentclass{llncs}

\usepackage{epsfig}
\usepackage{times}
\usepackage{gensymb}
\usepackage[english]{babel}
\usepackage{graphicx}
\usepackage{amssymb}
\usepackage{compress}
\usepackage{multirow}

\title{Strip Planarity Testing of Embedded Planar Graphs}

\date{}

\author{Patrizio Angelini$^1$, Giordano Da Lozzo$^1$, Giuseppe Di
Battista$^1$, Fabrizio Frati$^2$
\institute{
$1$ Dipartimento di Ingegneria, Roma Tre University, Italy\\
\email{\{angelini,dalozzo,gdb\}@dia.uniroma3.it}\\
$2$ School of Information Technologies, The University of Sydney, Australia\\
\email{brillo@it.usyd.edu.au}
}}

\newcommand{\remove}[1]{}

\renewcommand{\int}{int}

\newtheorem{claimx}{Claim}

\renewenvironment{proof}
{{\bf Proof:}}{\hspace*{\fill}$\Box$\par\vspace{2mm}}

\newenvironment{proof-strict-quasijagged}
{\vspace{2mm}{\bf Proof of Claim~\ref{cl:strict-quasijagged}:}}{\hspace*{\fill}$\Box$\par\vspace{2mm}}

\newenvironment{proofsketch}
{{\bf Proof sketch:}}{\hspace*{\fill}$\Box$\par\vspace{2mm}}

\begin{document}

\maketitle

\begin{abstract}
In this paper we introduce and study the {\em strip planarity testing} problem, which takes as an input a planar graph $G(V,E)$ and a function $\gamma:V \rightarrow \{1,2,\dots,k\}$ and asks whether a planar drawing of $G$ exists such that each edge is monotone in the $y$-direction and, for any $u,v\in V$ with $\gamma(u)<\gamma(v)$, it holds $y(u)<y(v)$. The problem has strong relationships with some of the most deeply studied variants of the planarity testing problem, such as {\em clustered planarity}, {\em upward planarity}, and {\em level planarity}. We show that the problem is polynomial-time solvable if $G$ has a fixed planar embedding.
\end{abstract}

\section{Introduction} \label{se:introduction}

Testing the planarity of a given graph is one of the oldest and most deeply
investigated problems in algorithmic graph theory. A celebrated result of
Hopcroft and Tarjan~\cite{ht-ept-74} states that the planarity testing problem
is solvable in linear time.

A number of interesting variants of the planarity testing problem have been
considered in the literature~\cite{s-ttp-13}. Such variants mainly focus on testing, for a given
planar graph $G$, the existence of a planar drawing of $G$ satisfying certain
{\em constraints}. For example the {\em partial embedding planarity}
problem~\cite{adfj-tppeg-10,jkr-kttppeg-11} asks whether a plane drawing $\cal
G$ of a given planar graph $G$ exists in which the drawing of a subgraph $H$ of
$G$ in $\cal G$ coincides with a given drawing $\cal H$ of $H$. {\em Clustered
planarity testing}~\cite{df-ectefcgsf-j09,jkkpsv-cpscce-j09}, {\em upward
planarity testing}~\cite{bdlm-udtg-94,GargT01,HuttonL96}, {\em
level planarity testing}~\cite{jlp-lptlt-98}, {\em embedding
constraints planarity testing}~\cite{gkm-ptoei-08}, {\em radial level planarity
testing}~\cite{bbf-rlpt-05}, and {\em clustered level planarity
testing}~\cite{fb-clp-04} are further examples of problems falling in this
category.

In this paper we introduce and study the {\em strip planarity testing} problem, which is defined as follows. The input of the problem consists of a planar graph $G(V,E)$ and of a function $\gamma:V \rightarrow \{1,2,\dots,k\}$. The problem asks whether a {\em strip planar} drawing of $(G,\gamma)$ exists, i.e. a planar drawing of $G$ such that each edge is monotone in the $y$-direction and, for any $u,v\in V$ with $\gamma(u)<\gamma(v)$, it holds $y(u)<y(v)$. The name ``strip'' planarity comes from the fact that, if a strip planar drawing $\Gamma$ of $(G,\gamma)$ exists, then $k$ disjoint horizontal strips $\gamma_1,\gamma_2,\dots,\gamma_k$ can be drawn in $\Gamma$ so that $\gamma_i$ lies below $\gamma_{i+1}$, for $1\leq i\leq k-1$, and so that $\gamma_i$ contains a vertex $x$ of $G$ if and only if $\gamma(x)=i$, for $1\leq i\leq k$. It is not difficult to argue that strips $\gamma_1,\gamma_2,\dots,\gamma_k$ can be given as part of the input, and the problem is to decide whether $G$ can be planarly drawn so that each edge is monotone in the $y$-direction and each vertex $x$ of $G$ with $\gamma(x)=i$ lies in the strip $\gamma_i$. That is, arbitrarily predetermining the placement of the strips does not alter the possibility of constructing a strip planar drawing of $(G,\gamma)$.

Before presenting our result, we discuss the strong relationships of the strip planarity testing problem with three famous graph drawing problems.

{\bf Strip planarity and clustered planarity.} The {\em $c$-planarity testing}
problem takes as an input a {\em clustered graph} $C(G,T)$, that is a planar
graph $G$ together with a rooted tree $T$, whose leaves are the vertices of $G$.
Each internal node $\mu$ of $T$ is called {\em cluster} and is associated with
the set $V_{\mu}$ of vertices of $G$ in the subtree of $T$ rooted at $\mu$. The
problem asks whether a {\em $c$-planar drawing} exists, that is a planar drawing
of $G$ together with a drawing of each cluster $\mu\in T$ as a simple closed
region $R_\mu$ so that: (i) if $v\in V_\mu$, then $v\in R_{\mu}$; (ii) if $V_\nu \subset V_\mu$, then $R_{\nu}\subset R_{\mu}$; (iii) if $V_\nu \cap V_\mu = \emptyset$, then $R_{\nu}\cap  R_{\mu}=\emptyset$; and (iv) each edge of $G$ intersects the border of $R_\mu$ at most once. Determining the time
complexity of testing the $c$-planarity of a given clustered graph is a
long-standing open problem. See~\cite{df-ectefcgsf-j09,jkkpsv-cpscce-j09} for
two recent papers on the topic. An instance $(G,\gamma)$ of the strip planarity
testing problem naturally defines a clustered graph $C(G,T)$, where $T$ consists
of a root having $k$ children $\mu_1,\dots,\mu_k$ and, for every $1\leq j\leq
k$, cluster $\mu_j$ contains every vertex $x$ of $G$ such that $\gamma(x)=j$.
The $c$-planarity of $C(G,T)$ is a necessary condition for the strip planarity
of $(G,\gamma)$, since suitably bounding the strips in a strip planar drawing of
$(G,\gamma)$ provides a $c$-planar drawing of $C(G,T)$. However, the
$c$-planarity of $C(G,T)$ is not sufficient for the strip planarity of
$(G,\gamma)$ (see Fig.~\ref{fig:strip-clustered}(a)). It turns out that strip planarity testing {\em coincides} with a special case of a problem opened by Cortese et al.~\cite{cdpp-cccc-05,cdbpp-ecpg-09} and related to $c$-planarity testing. The problem asks whether a graph $G$ can be planarly embedded ``inside'' an host graph $H$, which can be thought as having ``fat'' vertices and edges, with each vertex and edge of $G$ drawn inside a prescribed vertex and a prescribed edge of $H$, respectively. It is easy to see that the strip planarity testing problem coincides with this problem in the case in which $H$ is a path.

\begin{figure}[tb]
\begin{center}
\begin{tabular}{c c}
\mbox{\includegraphics[scale=0.5]{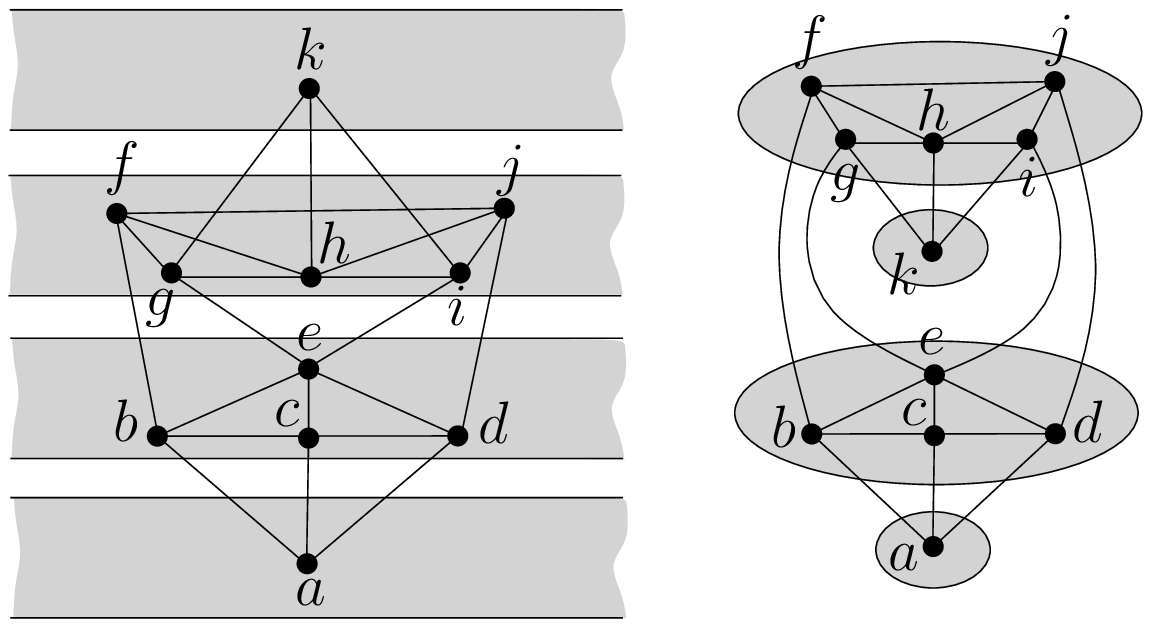}} \hspace{10mm}
&
\mbox{\includegraphics[scale=0.5]{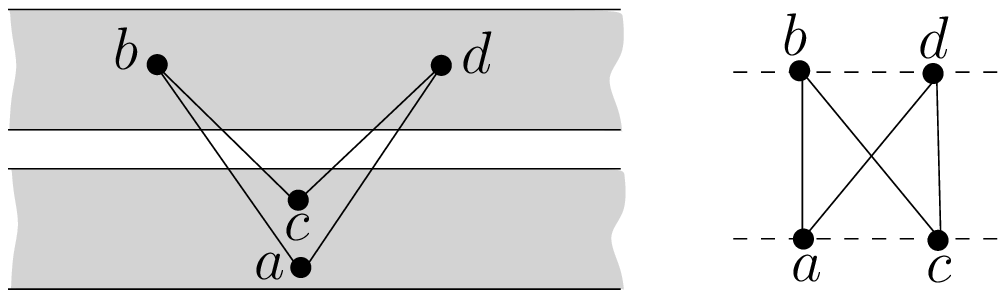}} \\
(a) \hspace{10mm}& (b)
\end{tabular}
\caption{(a) A negative instance $(G,\gamma)$ of the strip planarity testing problem whose associated clustered graph $C(G,T)$ is $c$-planar. (b) A positive instance $(G,\gamma)$ of the strip planarity testing problem that is not level planar.}
\label{fig:strip-clustered}
\end{center}
\vspace{-5mm}
\end{figure}

{\bf Strip planarity and level planarity.} The {\em level planarity testing}
problem takes as an input a planar graph $G(V,E)$ and a function $\gamma:V
\rightarrow \{1,2,\dots,k\}$ and asks whether a planar drawing of $G$ exists
such that each edge is monotone in the $y$-direction and each vertex $u \in V$
is drawn on the horizontal line $y=\gamma(u)$. The level planarity testing (and
embedding) problem is known to be solvable in linear
time~\cite{jlp-lptlt-98}, although a sequence of incomplete
characterizations by forbidden subgraphs~\cite{fk-mlnpt-07,hkl-clpg-04} (see
also~\cite{efk-clptmp-10}) has revealed that the problem is not yet fully
understood. The similarity of the level planarity testing problem with the strip
planarity testing problem is evident: They have the same input, they both
require planar drawings with $y$-monotone edges, and they both constrain the
vertices to lie in specific regions of the plane; they only differ for the fact
that such regions are horizontal lines in one case, and horizontal strips in the
other one. Clearly the level planarity of an instance $(G,\gamma)$ is a
sufficient condition for the strip planarity of $(G,\gamma)$, as a level planar
drawing is also a strip planar drawing. However, it is easy to construct instances $(G,\gamma)$ that are strip planar and yet not level planar, even if we require that the instances are {\em strict}, i.e., no edge $(u,v)$ is such that $\gamma(u)=\gamma(v)$.
See Fig.~\ref{fig:strip-clustered}(b). Also, the approach of~\cite{jlp-lptlt-98} seems to be not applicable to test the strip planarity of a graph. Namely, J{\"u}nger et al.~\cite{jlp-lptlt-98} visit the instance $(G,\gamma)$ one level at a time, representing with a PQ-tree~\cite{bl-tcop-76} the possible orderings of the vertices in level $i$ that are consistent with a level planar embedding of the subgraph of $G$ induced by levels $\{1,2,\dots,i\}$. However, when visiting an instance $(G,\gamma)$ of the strip planarity testing problem one strip at a time, PQ-trees seem to be not powerful enough to represent the possible orderings of the vertices in strip $i$ that are consistent with a strip planar embedding of the subgraph of $G$ induced by strips $\{1,2,\dots,i\}$.

{\bf Strip planarity and upward planarity.} The {\em upward planarity testing}
problem asks whether a given directed graph $\overrightarrow{G}$ admits an {\em
upward planar drawing}, i.e., a drawing which is planar and such that each edge
is represented by a curve monotonically increasing in the $y$-direction,
according to its orientation. Testing the upward planarity of a directed graph
$\overrightarrow{G}$ is an $\cal{NP}$-hard problem~\cite{GargT01}, however it is
polynomial-time solvable, e.g., if $\overrightarrow{G}$ has a fixed
embedding~\cite{bdlm-udtg-94}, or if it has a single-source~\cite{HuttonL96}. A strict instance $(G,\gamma)$ of
the strip planarity testing problem naturally defines a directed graph
$\overrightarrow{G}$, by directing an edge $(u,v)$ of $G$ from $u$ to $v$ if
$\gamma(u)<\gamma(v)$. It is easy to argue that the upward planarity of
$\overrightarrow{G}$ is a necessary and not sufficient condition for the strip
planarity of $(G,\gamma)$ (see Fig.s~\ref{fig:non-strip-planar-upward}(a) and~\ref{fig:non-strip-planar-upward}(b)). Roughly speaking, in an upward planar drawing different parts of the
graph are free to ``nest'' one into the other, while in a strip planar drawing,
such a nesting is only allowed if coherent with the strip assignment.

\begin{figure}[tb]
\begin{center}
\begin{tabular}{c c}
\mbox{\includegraphics[scale=0.48]{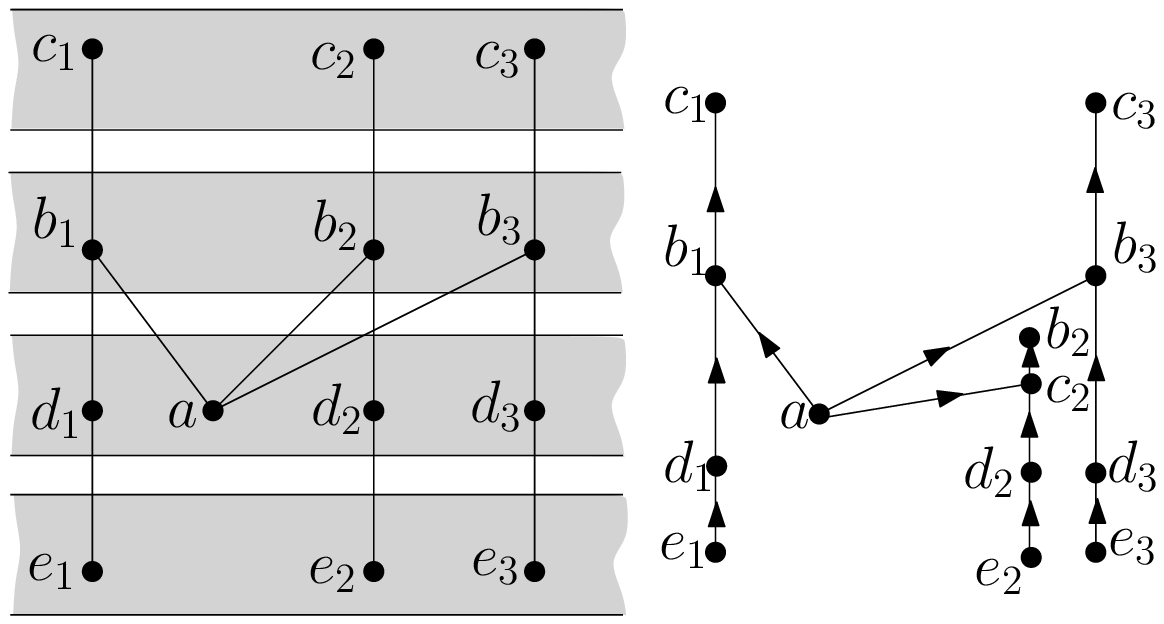}} \hspace{5mm}
&
\mbox{\includegraphics[scale=0.48]{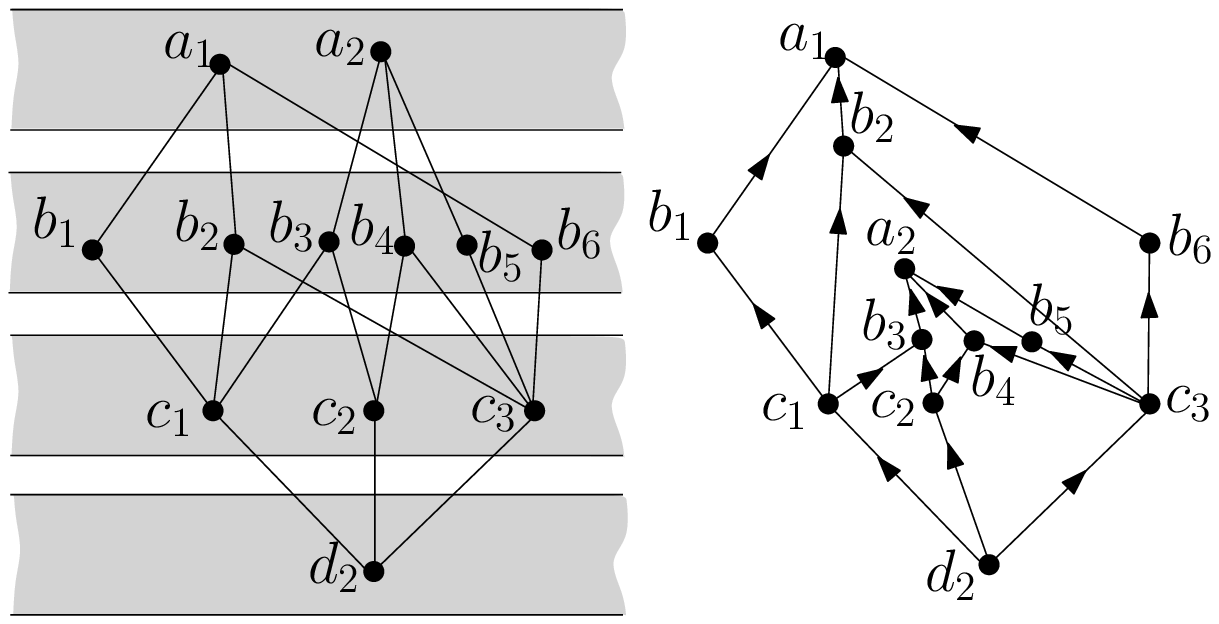}} \\
(a) \hspace{5mm}& (b)
\end{tabular}
\caption{Two negative instances $(G_1,\gamma_1)$ (a) and $(G_2,\gamma_2)$ (b) whose associated directed graphs are upward planar, where $G_1$ is a tree and $G_2$ is a subdivision of a triconnected plane graph.}
\label{fig:non-strip-planar-upward}
\end{center}
\vspace{-5mm}
\end{figure}

In this paper, we show that the strip planarity testing problem is polynomial-time solvable for planar graphs with a fixed planar embedding. Our approach consists of performing a sequence of modifications to the input instance $(G,\gamma)$ (such modifications consist mainly of insertions of graphs inside the faces of $G$) that ensure that the instance satisfies progressively stronger constraints while not altering its strip planarity. Eventually, the strip planarity of $(G,\gamma)$ becomes equivalent to the upward planarity of its associated directed graph, which can be tested in polynomial time.

The rest of the paper is organized as follows. In Section~\ref{se:preliminaries} we present some preliminaries; in Section~\ref{se:theorems} we prove our result; finally, in Section~\ref{se:conclusions} we conclude and present open problems. Because of space limitations, the proofs are sketched or omitted. 

\section{Preliminaries} \label{se:preliminaries}

In this section we present some definitions and preliminaries.

A \emph{drawing} of a graph is a mapping of each vertex to a distinct point of the plane and of each edge to a Jordan curve between the endpoints of the edge. A \emph{planar drawing} is such that no two edges intersect except, possibly, at common endpoints. A planar drawing of a graph determines a circular ordering of the edges incident to each vertex. Two drawings of the same graph are \emph{equivalent} if they determine the same circular orderings around each vertex. A \emph{planar embedding} (or \emph{combinatorial embedding}) is an equivalence class of planar drawings. A planar drawing partitions the plane into topologically connected regions, called {\em faces}. The unbounded face is the {\em outer face}. Two planar drawings with the same combinatorial embedding have the same faces. However, such drawings could still differ for their outer faces. A {\em plane embedding} of a graph $G$ is a planar embedding of $G$ together with a choice for its outer face.

In this paper we will show how to test in polynomial time whether a graph with a prescribed {\em plane} embedding is strip planar. Since a graph with a fixed combinatorial embedding has $O(n)$ choices for the outer face, this implies that testing the strip planarity of a graph with a prescribed {\em combinatorial} embedding is also a polynomial-time solvable problem. In the reminder of the paper, we will assume all the considered graphs to have a prescribed plane embedding, even when not explicitly mentioned.

For the sake of simplicity of description, in the following we assume that the considered plane graphs are {\em $2$-connected}, unless otherwise specified. We will sketch in the conclusions how to extend our results to simply-connected and even non-connected plane graphs. We now define some concepts related to strip planarity.

An instance $(G,\gamma)$ of the strip planarity testing problem is {\em strict} if it contains no intra-strip edge, where an edge $(u,v)$ is {\em intra-strip} f $\gamma(u)=\gamma(v)$. An instance $(G,\gamma)$ of strip planarity is {\em proper}
if, for every edge $(u,v)$ of $G$, it holds $\gamma(v)-1\leq\gamma(u) \leq
\gamma(v)+1$. Given any non-proper instance of strip planarity, one can replace
every edge $(u,v)$ such that $\gamma(u)=\gamma(v)+j$, for some $j\geq 2$, with a
path $(v=u_1,u_2,\dots,u_{j+1}=u)$ such that $\gamma(u_{i+1})=\gamma(u_i)+1$,
for every $1\leq i\leq j$, thus obtaining a proper instance $(G',\gamma')$ of
the strip planarity testing problem. It is easy to argue that
$(G,\gamma)$ is strip planar if and only if $(G',\gamma')$ is strip planar. In
the following, we will assume all the considered instances of the strip
planarity testing problem to be proper, even when not explicitly mentioned. 

Let $(G,\gamma)$ be an instance of the strip planarity testing problem. A path $(u_1,\dots,u_j)$ in $G$ is {\em monotone} if $\gamma(u_i)=\gamma(u_{i-1})+1$, for every $2\leq i\leq j$. For any face $f$ in $G$, we denote by $C_f$ the simple cycle delimiting the border of $f$. Let $f$ be a face of $G$, let $u$ be a vertex incident to $f$, and let $v$ and
$z$ be the two neighbors of $u$ on $C_f$. We say that $u$ is a {\em local
minimum} for $f$ if $\gamma(v)=\gamma(z)=\gamma(u)+1$, and it
is a {\em local maximum} for $f$ if $\gamma(v)=\gamma(z)=\gamma(u)-1$. Also, we say that $u$ is a {\em global minimum} for $f$ (a {\em global maximum} for $f$)
if $\gamma(w)\geq \gamma(u)$ (resp. $\gamma(w)\leq \gamma(u)$), for every vertex $w$ incident to $f$. A global minimum $u_m$ and a global maximum
$u_M$ for a face $f$ are {\em consecutive} in $f$ if no global minimum and no global maximum exists in one of the two paths connecting $u_m$ and $u_M$ in
$C_f$. A local minimum $u_m$ and a local maximum $u_M$ for a face $f$ are {\em
visible} if one of the paths $P$ connecting $u_m$ and $u_M$ in $C_f$ is such
that, for every vertex $u$ of $P$, it holds $\gamma(u_m) < \gamma(u) <
\gamma(u_M)$.

\begin{definition}
An instance $(G,\gamma)$ of the strip planarity problem is {\em quasi-jagged} if it is strict and if, for every face $f$ of $G$ and for any two
visible local minimum $u_m$ and local maximum $u_M$ for $f$, one of the
two paths connecting $u_m$ and $u_M$ in $C_f$ is monotone.
\end{definition}

\begin{definition}
An instance $(G,\gamma)$ of the strip planarity problem is {\em jagged}
if it is strict and if, for every face $f$ of $G$, any local minimum for $f$ is
a global minimum for $f$, and every local maximum for $f$ is a global maximum
for $f$.
\end{definition}


\section{How To Test Strip Planarity} \label{se:theorems}

In this section we show an algorithm to test strip planarity.

\remove{In
Section~\ref{subse:general-strict} we show how to reduce in polynomial time a
general instance to an equivalent strict
instance. In Section~\ref{subse:strict-quasijagged} we show how to reduce in
polynomial time a strict instance to an
equivalent quasi-jagged instance. In Section~\ref{subse:quasijagged-jagged} we
show how to reduce in polynomial time a quasi-jagged instance to a jagged instance. Finally, in
Section~\ref{subse:jagged-upward} we show that testing the strip planarity of a
jagged instance is equivalent to test the upward planarity of the associated
directed graph, which can be done in polynomial time.}

\subsection{From a General Instance to a Strict Instance}\label{subse:general-strict}

In this section we show how to reduce a general instance of the strip planarity
testing problem  to an equivalent strict instance.

\begin{lemma} \label{le:general-strict}
Let $(G,\gamma)$ be an instance of the strip planarity testing problem. Then, there exists a polynomial-time algorithm that either constructs an equivalent strict instance $(G^*,\gamma^*)$ or decides that $(G,\gamma)$ is not strip planar.
\end{lemma}

Consider any intra-strip edge $(u,v)$ in $G$, if it exists. We distinguish two cases.

In {\em Case 1}, $(u,v)$ is an edge of a $3$-cycle $(u,v,z)$ that contains vertices in its interior in $G$. Observe that, $\gamma(u)-1\leq \gamma(z) \leq \gamma(u)+1$. Denote by $G'$ the plane subgraph of $G$ induced by the vertices lying outside cycle $(u,v,z)$ together with $u$, $v$, and $z$ (this graph might coincide with cycle $(u,v,z)$ if such a cycle delimits the outer face of $G$); also, denote by $G''$ the plane subgraph of $G$ induced by the vertices lying inside cycle $(u,v,z)$ together with $u$, $v$, and $z$. Also, let $\gamma'(x)=\gamma(x)$, for every vertex $x$ in $G'$, and let $\gamma''(x)=\gamma(x)$, for every vertex $x$ in $G''$. We have the following:

\begin{claimx} \label{cl:general-strict-1}
$(G,\gamma)$ is strip planar if and only if $(G',\gamma')$ and $(G'',\gamma'')$ are both strip planar.
\end{claimx}
\begin{proof}
The necessity of the conditions is trivial, given that $G'$ and $G''$ are
subgraphs of $G$, that $\gamma(x)=\gamma'(x)$, for every vertex $x$ of $G'$, and
that $\gamma(x)=\gamma''(x)$, for every vertex $x$ of $G''$.

The sufficiency of the conditions is easily proved as follows. Suppose that
$(G',\gamma')$ and $(G'',\gamma'')$ admit strip planar drawings $\Gamma'$ and
$\Gamma''$, respectively. Scale $\Gamma''$ so that it fits inside the drawing of
cycle $(u,v,z)$ in $\Gamma'$. If $\gamma''(z) = \gamma''(u)$, then suitably
stretch the edges of $G''$ in $\Gamma''$ so that: (i) the drawing of cycle
$(u,v,z)$ in $\Gamma''$ coincides with the drawing of cycle $(u,v,z)$ in
$\Gamma'$ and (ii) no two edges in $\Gamma''$ cross. Then, the drawing $\Gamma$
obtained by gluing $\Gamma'$ and $\Gamma''$ along cycle $(u,v,z)$ is a strip
planar drawing of $(G,\gamma)$. If $\gamma''(z) = \gamma''(u)-1$ (the case in
which $\gamma''(z) = \gamma''(u)+1$ is analogous), then suitably stretch the
edges of $G''$ in $\Gamma''$ so that: (i) the drawing of cycle $(u,v,z)$ in
$\Gamma''$ coincides with the drawing of cycle $(u,v,z)$ in $\Gamma'$, (ii) no
two edges in $\Gamma''$ cross, and (iii) each vertex $x$ of $G''$ such that
$\gamma''(x)=\gamma''(u)$ lies in the strip associated with $\gamma''(u)$ and
each vertex $x$ of $G''$ such that $\gamma''(x)=\gamma''(z)$ lies in the strip
associated with $\gamma''(z)$. Then, the drawing $\Gamma$ obtained by gluing
$\Gamma'$ and $\Gamma''$ along cycle $(u,v,z)$ is a strip planar drawing of
$(G,\gamma)$.
\end{proof}

The strip planarity of $(G'',\gamma'')$ can be tested in linear time as follows.

If $\gamma''(z) = \gamma''(u)$, then $(G'',\gamma'')$ is strip planar if and only if $\gamma''(x)=\gamma''(u)$ for every vertex $x$ of $G''$ (such a condition can clearly be tested in linear time). For the necessity, $3$-cycle $(u,v,z)$ is entirely drawn in $\gamma''(u)$, hence all the internal vertices of $G''$ have to be drawn inside $\gamma''(u)$ as well. For the sufficiency, $G''$ has a plane embedding by assumption, hence any planar $y$-monotone drawing (e.g. a straight-line drawing where no two vertices have the same $y$-coordinate) respecting such an embedding and contained in $\gamma''(u)$ is a strip planar drawing of $(G'',\gamma'')$.

If $\gamma''(z) = \gamma''(u)-1$ (the case in which $\gamma''(z) = \gamma''(u)+1$ is analogous), then we argue as follows: First, a clustered graph $C(G'',T)$ can be defined such that $T$ consists of two clusters $\mu$ and $\nu$, respectively containing every vertex $x$ of $G''$ such that $\gamma''(x)=\gamma''(u)-1$, and every vertex $x$ of $G''$ such that $\gamma''(x)=\gamma''(u)$. We show that $(G'',\gamma'')$ is strip planar if and only if $C(G'',T)$ is $c$-planar. For the necessity, it suffices to observe that a strip planar drawing of $(G'',\gamma'')$ is also a $c$-planar drawing of $C(G'',T)$. For the sufficiency, if $C(G'',T)$ admits a $c$-planar drawing, then it also admits a $c$-planar {\em straight-line} drawing $\Gamma(C)$ in which the regions $R(\mu)$ and $R(\nu)$ representing $\mu$ and $\nu$, respectively, are {\em convex}~\cite{afk-srdcg-11,EadesFLN06}. Assuming w.l.o.g. up to a rotation of $\Gamma(C)$ that $R(\mu)$ and $R(\nu)$ can be separated by a horizontal line, we have that disjoint horizontal strips can be drawn containing $R(\mu)$ and $R(\nu)$. Slightly perturbing the positions of the vertices so that no two of them have the same $y$-coordinate ensures that the the edges are $y$-monotone, thus resulting in a strip planar drawing of $(G'',\gamma'')$. Finally, the $c$-planarity of a clustered graph containing two clusters can be decided in linear time, as independently proved by Biedl et al.~\cite{bkm-dpp-98} and by Hong and Nagamochi~\cite{hn-tpbecp-09}.

In {\em Case 2}, a $3$-cycle $(u,v,z)$ exists that contains no vertices in its
interior in $G$. Then, {\em contract} $(u,v)$, that is, identify $u$ and $v$ to
be the same vertex $w$, whose incident edges are all the edges incident to $u$
and $v$, except for $(u,v)$; the clockwise order of the edges incident to $w$
is: All the edges that used to be incident to $u$ in the same clockwise order
starting at $(u,v)$, and then all the edges that used to be incident to $v$ in
the same clockwise order starting at $(v,u)$. Denote by $G'$ the resulting
graph. Since $G$ is plane, $G'$ is plane; since $G$ contains no $3$-cycle
$(u,v,z)$ that contains vertices in its interior, $G'$ is simple. Let
$\gamma'(x)=\gamma(x)$, for every vertex $x\neq u,v$ in $G$, and let
$\gamma'(w)=\gamma(u)$. We have the following.

\begin{claimx} \label{cl:general-strict-2}
$(G',\gamma')$ is strip planar if and only if $(G,\gamma)$ is strip planar.
\end{claimx}
\begin{proof}
For the necessity, consider any strip planar drawing $\Gamma$ of $(G,\gamma)$
(see Fig.~\ref{fig:claim2}(a)). Denote by $p_1,p_2,\dots,p_h$ and by
$q_1,q_2,\dots,q_l$ the left-to-right order of the intersection points of the
edges of $G$ with the lines delimiting strip $\gamma(u)$ from the top and from
the bottom, respectively. Insert dummy vertices at points $p_1,p_2,\dots,p_h$
and $q_1,q_2,\dots,q_l$. Each of such vertices splits an edge of $G$ into two
dummy edges, one lying inside $\gamma(u)$ and one not. Insert dummy edges
$(p_1,q_1)$, $(p_h,q_l)$, $(p_i,p_{i+1})$, for $1\leq i \leq h-1$, and
$(q_i,q_{i+1})$, for $1\leq i \leq l-1$, in $\gamma(u)$.

Contract edge $(u,v)$ into a single vertex $w$. Such a contraction does not
introduce multiple edges, given that no separating $3$-cycle $(u,v,x)$ exists in
$G$, by assumption. Triangulate the internal faces of the resulting plane graph
by inserting dummy vertices and edges, in such a way that no edge connects two
vertices $p_i$ and $p_j$ with $j\geq i+2$, and no edge connects two vertices
$q_i$ and $q_j$ with $j\geq i+2$ (see Fig.~\ref{fig:claim2}(b)). Denote by $T_d$
the resulting internally-triangulated simple plane graph.

Construct a convex straight-line drawing of $T_d$ in which vertices
$p_1,p_2,\dots,p_h$ and $q_1,q_2,\dots,q_l$ have the same positions they have in
$\Gamma$ (see Fig.~\ref{fig:claim2}(c)). Such a drawing always
exists~\cite{cyn-lacdp-84}. Slightly perturb the positions of the vertices
different from $p_1,p_2,\dots,p_h$ and $q_1,q_2,\dots,q_l$, so that no two of
them have the same $y$-coordinate. As a consequence, the edges of $T_d$
different from $(p_i,p_{i+1})$, for $1\leq i \leq h-1$, and $(q_i,q_{i+1})$, for
$1\leq i \leq l-1$, are $y$-monotone curves. Removing the inserted dummy
vertices and edges results in a strip planar drawing of $(G',\gamma')$ (see
Fig.~\ref{fig:claim2}(d)).

\begin{figure}[htb]
\begin{center}
\begin{tabular}{c c}
\mbox{\includegraphics[scale=0.33]{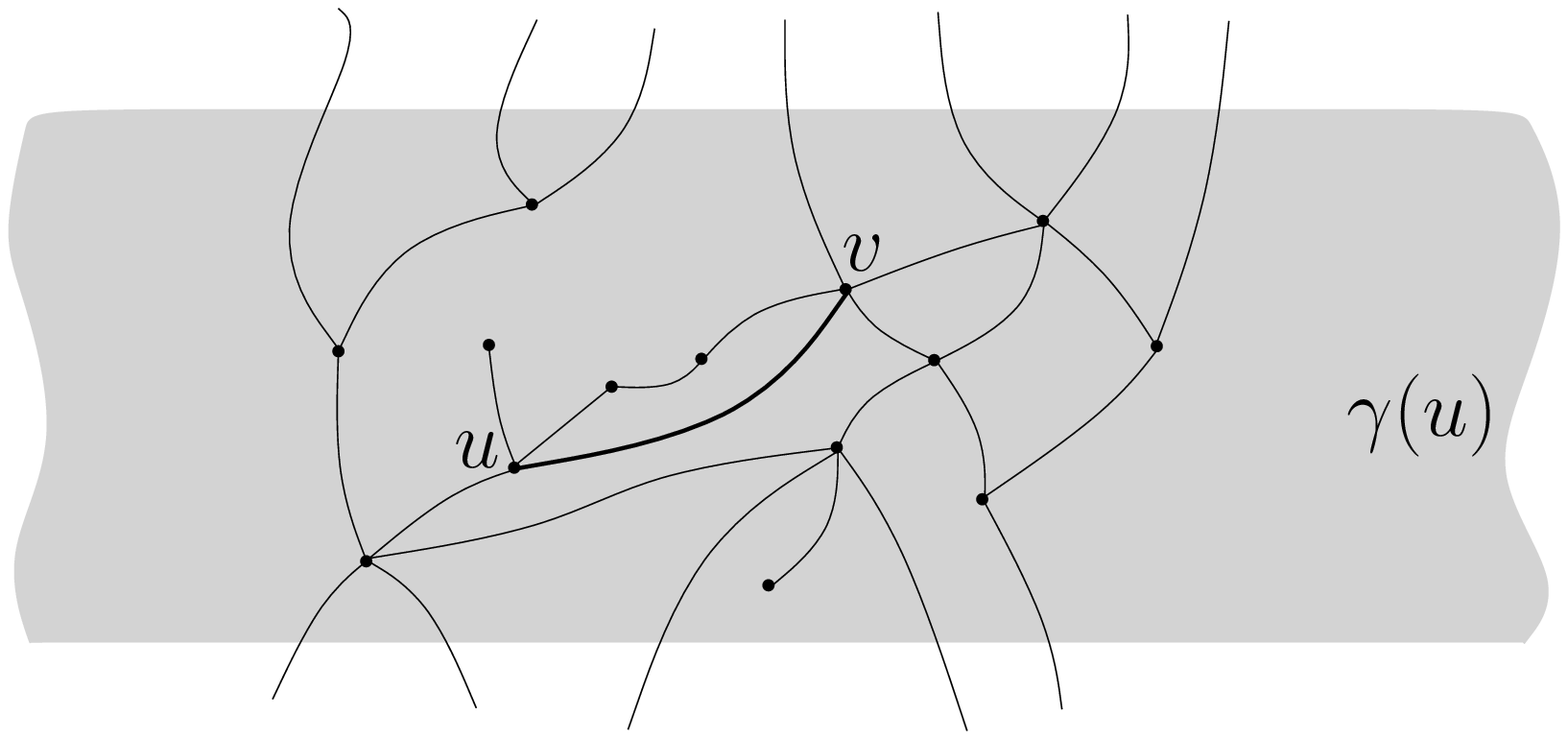}} \hspace{10mm} &
\mbox{\includegraphics[scale=0.33]{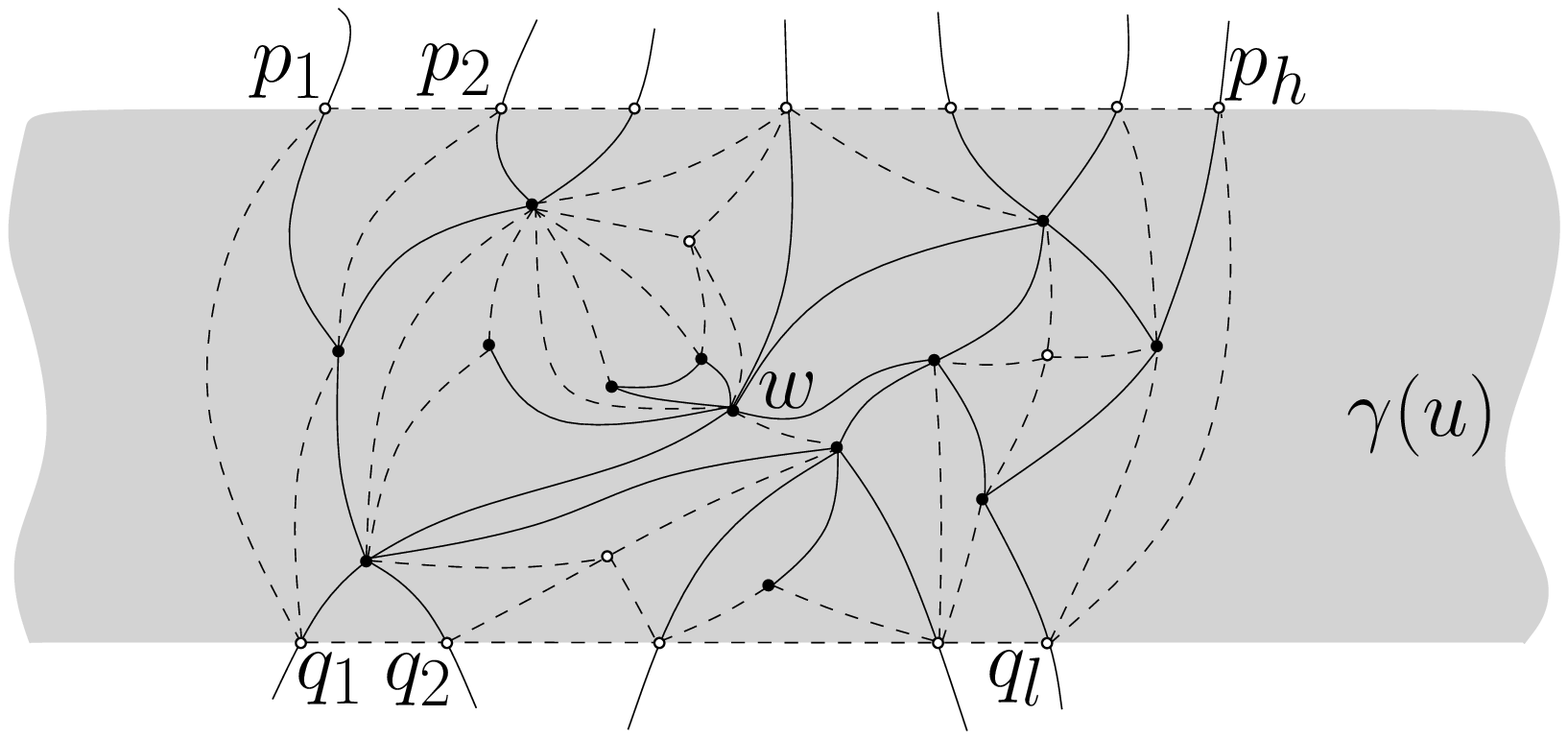}} \\
(a) \hspace{10mm}& (b)\\
\mbox{\includegraphics[scale=0.33]{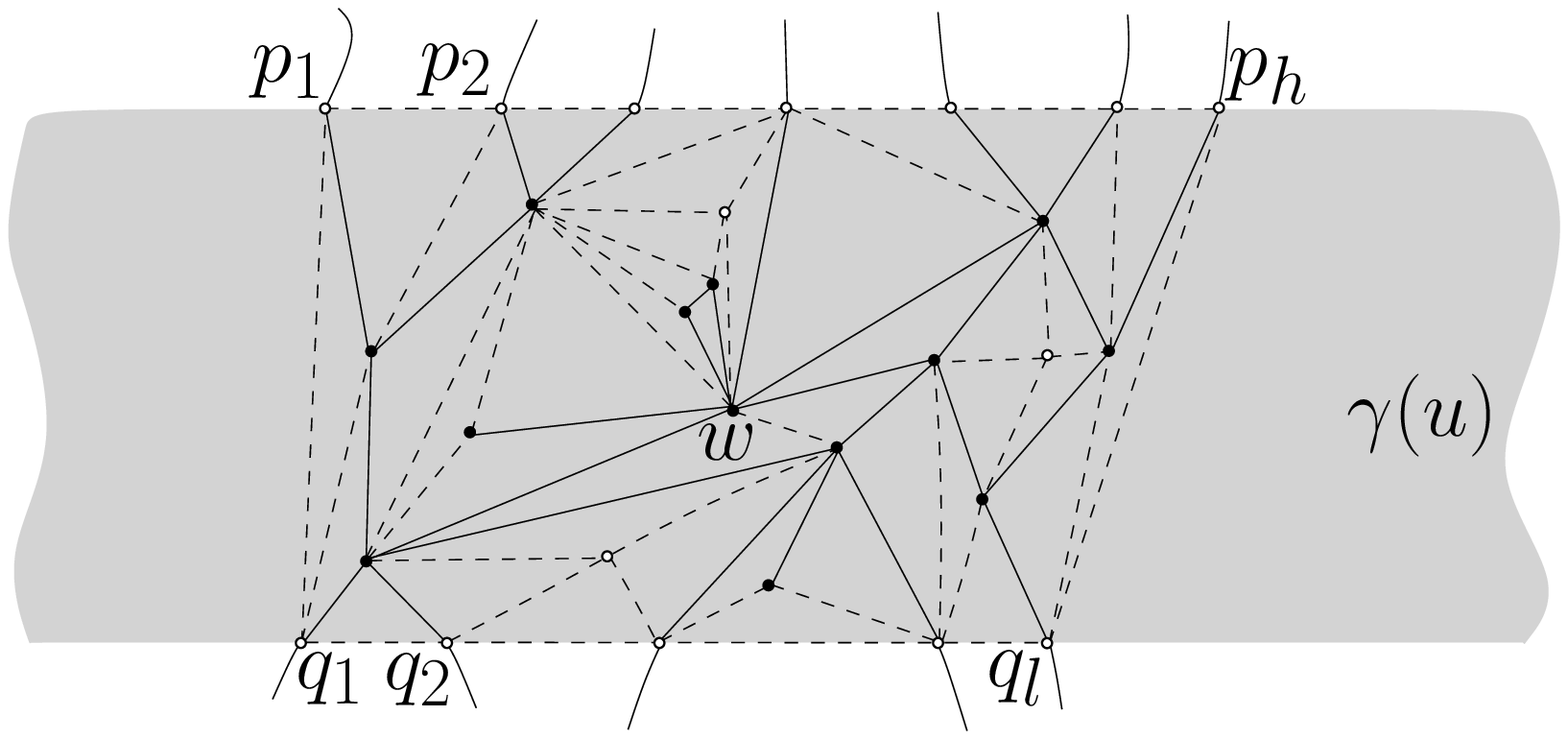}} \hspace{10mm} &
\mbox{\includegraphics[scale=0.33]{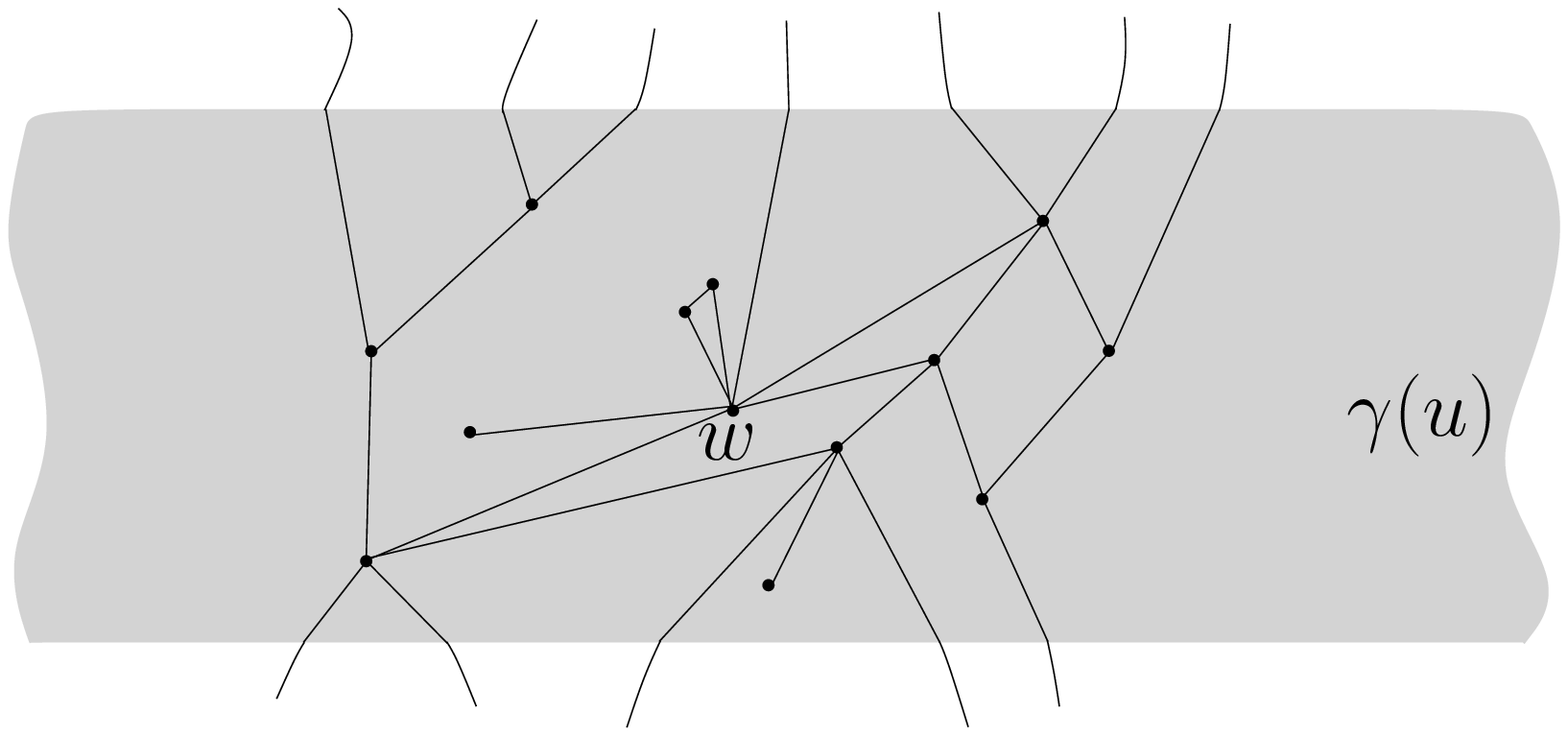}} \\
(c) \hspace{10mm}& (d)
\end{tabular}
\caption{(a) A strip planar drawing $\Gamma$ of $(G,\gamma)$. (b) Modifications
performed on the part of $G$ inside $\gamma(u)$, resulting in a
internally-triangulated simple plane graph $T_d$. (c) A convex straight-line
drawing of $T_d$. (d) A strip planar drawing of $(G',\gamma')$.}
\label{fig:claim2}
\end{center}
\end{figure}

\begin{figure}[htb]
\begin{center}
\begin{tabular}{c c c}
\mbox{\includegraphics[scale=0.6]{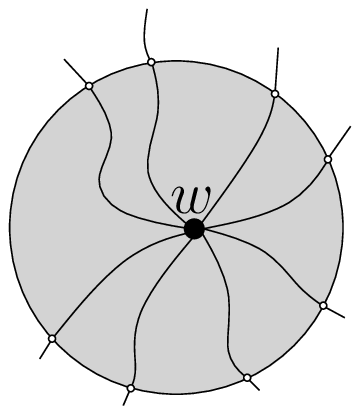}} \hspace{10mm} &
\mbox{\includegraphics[scale=0.6]{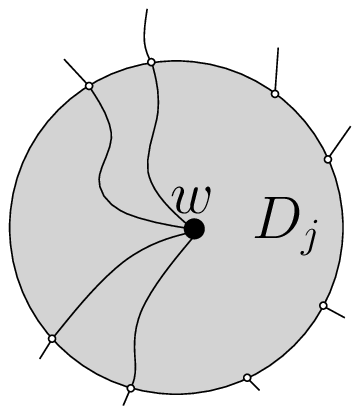}} \hspace{10mm} &
\mbox{\includegraphics[scale=0.6]{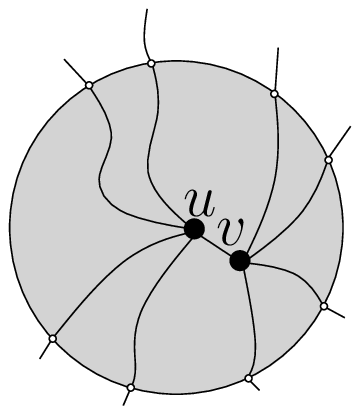}} \\
(a) \hspace{10mm}& (b) \hspace{10mm}& (c)
\end{tabular}
\caption{(a) A disk $D$ containing $w$. (b) Region $D_j$. (c) Drawing edge
$(u,v)$ and the edges incident to $v$ inside $D$.}
\label{fig:disk}
\end{center}
\end{figure}

For the sufficiency, consider any strip planar drawing $\Gamma'$ of
$(G',\gamma')$. Slightly perturb the positions of the vertices in $\Gamma'$, so
that no two vertices have the same $y$-coordinate. Consider a disk $D$
containing $w$, small enough so that it contains no vertex different from $w$,
and it contains no part of an edge that is not incident to $w$ (see
Fig.~\ref{fig:disk}(a)). Remove from the interior of $D$ the parts of the edges
incident to $w$ that correspond to edges incident to $v$. The edges still
incident to $w$ partition $D$ into regions $D_1,D_2,\dots,D_l$. At most one of
such regions, say $D_j$, used to contain edges incident to $w$ corresponding to
edges incident to $v$. In fact, all the edges incident to $w$ corresponding to
edges incident to $v$ appear consecutively around $w$ in $G'$ (see
Fig.~\ref{fig:disk}(b)). Insert a $y$-monotone curve incident to $w$ in $D_j$.
Let $v$ be the end-vertex of such a curve different from $w$. Rename $w$ to $u$.
Draw $y$-monotone curves connecting $v$ with the intersection points of the
border of $D_j$ with the edges incident to $w$ that used to lie inside $D_j$(see
Fig.~\ref{fig:disk}(c)). The resulting drawing is a strip planar drawing of
$(G,\gamma)$.
\end{proof}

Claims~\ref{cl:general-strict-1} and~\ref{cl:general-strict-2} imply
Lemma~\ref{le:general-strict}.
Namely, if $(G,\gamma)$ has no intra-strip edge, there is nothing to prove.
Otherwise, $(G,\gamma)$ has an intra-strip edge $(u,v)$, hence either Case~1 or
Case~2 applies. If Case~2 applies to $(G,\gamma)$, then an instance
$(G',\gamma')$ is obtained in linear time containing one less vertex
than $(G,\gamma)$. By Claim~\ref{cl:general-strict-2}, $(G',\gamma')$ is
equivalent to $(G,\gamma)$. Otherwise, Case~1 applies to $(G,\gamma)$. Then,
either the non-strip planarity of $(G,\gamma)$ is deduced (if $(G'',\gamma'')$
is not strip planar), or an instance $(G',\gamma')$ is obtained containing at
least one less vertex than $(G,\gamma)$  (if $(G'',\gamma'')$ is strip planar). By Claim~\ref{cl:general-strict-1},
$(G',\gamma')$ is equivalent to $(G,\gamma)$. The repetition of such an
argument either leads to conclude in polynomial time that $(G,\gamma)$ is not
strip planar, or leads to construct in polynomial time a strict instance
$(G^*,\gamma^*)$ of strip planarity equivalent to $(G,\gamma)$.

\subsection{From a Strict Instance to a Quasi-Jagged Instance}\label{subse:strict-quasijagged}

In this section we show how to reduce a strict instance of the strip planarity
testing problem to an equivalent quasi-jagged instance. Again, for the sake of simplicity of description, we assume that every considered instance $(G,\gamma)$ is $2$-connected.

\begin{lemma} \label{le:strict-quasijagged}
Let $(G,\gamma)$ be a strict instance of the strip planarity testing problem.
Then, there exists a polynomial-time algorithm that constructs an equivalent
quasi-jagged instance $(G^*,\gamma^*)$ of the strip planarity testing problem.
\end{lemma}

Consider any face $f$ of $G$ containing two visible local minimum and maximum
$u_m$ and $u_M$, respectively, such that no path connecting $u_m$ and $u_M$ in
$C_f$ is monotone. Insert a monotone path connecting $u_m$ and
$u_M$ inside $f$. Denote by $(G^+,\gamma^+)$ the resulting instance of the strip
planarity testing problem. We have the following claim:

\begin{claimx} \label{cl:strict-quasijagged}
$(G^+,\gamma^+)$ is strip planar if and only if $(G,\gamma)$ is strip planar.
\end{claimx}

\remove{
\begin{proofsketch}
The necessity is trivial. For the sufficiency, consider any strip planar
drawing $\Gamma$ of $(G,\gamma)$. Denote by $P$ the path connecting $u_m$ and
$u_M$ along $C_f$ and such that $\gamma(u_m)<\gamma(v)<\gamma(u_M)$ holds for
every internal vertex $v$ of $P$. Because of the existence of some parts of the
graph that ``intermingle'' with $P$, it might not be possible to draw a
$y$-monotone curve inside $f$ connecting $u_m$ and $u_M$ in $\Gamma$. Thus, a
part of $\Gamma$ has to be horizontally shrunk, so that it moves ``far away''
from $P$, thus allowing for the monotone path connecting $u_m$ and $u_M$ to be
drawn as a $y$-monotone curve inside $f$. This results in a strip planar drawing
of $(G^+,\gamma^+)$.
\end{proofsketch}
}
\begin{proof}
One direction of the equivalence is trivial, namely if $(G^+,\gamma^+)$ is strip
planar, then $(G,\gamma)$ is strip planar, since $G$ is a subgraph of $G^+$ and
$\gamma(v)=\gamma^+(v)$ for every vertex $v$ in $G$.

We prove the other direction. Consider a strip planar drawing $\Gamma$ of
$(G,\gamma)$. Slightly perturb the positions of the vertices in $\Gamma$ so that
no two of them have the same $y$-coordinate. Denote by $P$ and $Q$ the two paths
connecting $u_m$ and $u_M$ along $C_f$. Since $u_m$ and $u_M$ are visible local
minimum and maximum for $f$, it holds $\gamma(u_m)<\gamma(v)<\gamma(u_M)$ for
every internal vertex $v$ of $P$, or it holds
$\gamma(u_m)<\gamma(v)<\gamma(u_M)$ for every internal vertex $v$ of $Q$. Assume
that $\gamma(u_m)<\gamma(v)<\gamma(u_M)$ holds for every internal vertex $v$ of
$P$, the other case being analogous. We also assume w.l.o.g. that face $f$ is to
the right of $P$ when traversing such a path from $u_m$ to $u_M$. We modify
$\Gamma$, if necessary, while maintaining its strip planarity so that a
$y$-monotone curve $\cal C$ connecting $u_m$ and $u_M$ can be drawn inside $f$.

\begin{figure}[htb]
\begin{center}
\begin{tabular}{c c}
\mbox{\includegraphics[scale=0.5]{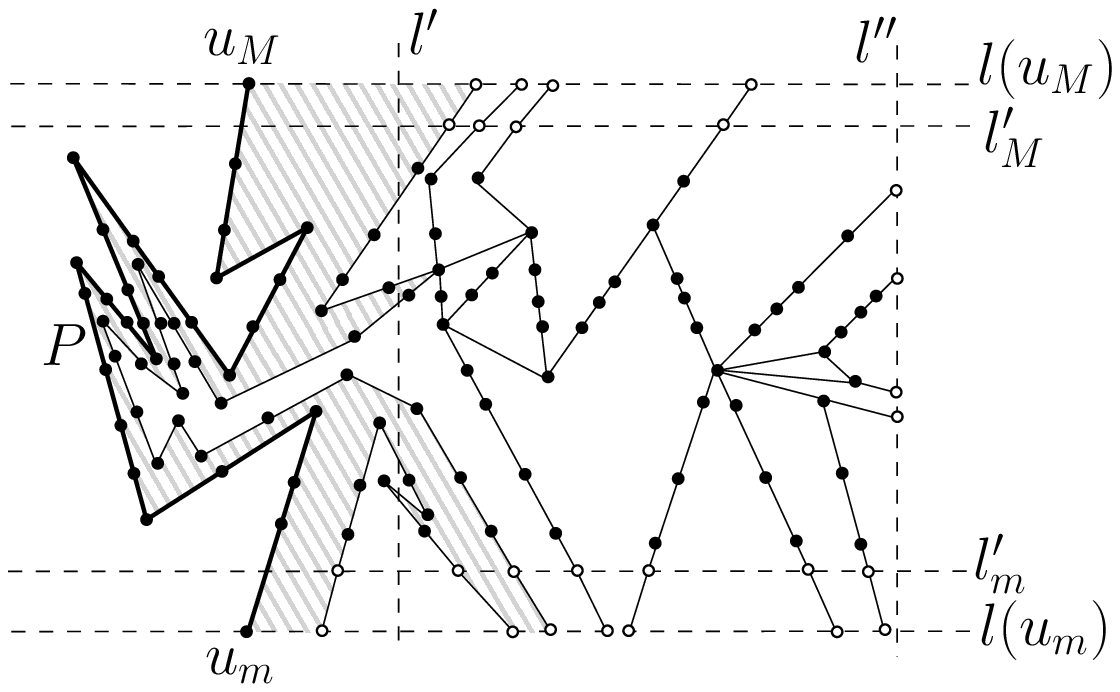}} \hspace{5mm} &
\mbox{\includegraphics[scale=0.5]{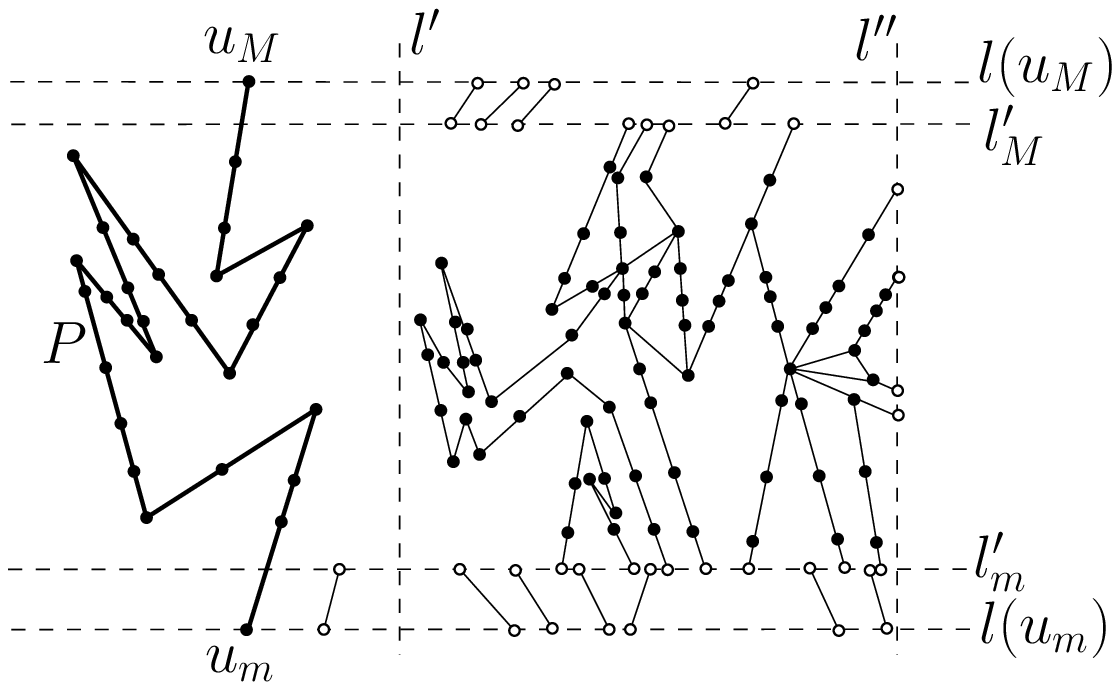}} \\
(a) \hspace{5mm}& (b)\\
\mbox{\includegraphics[scale=0.5]{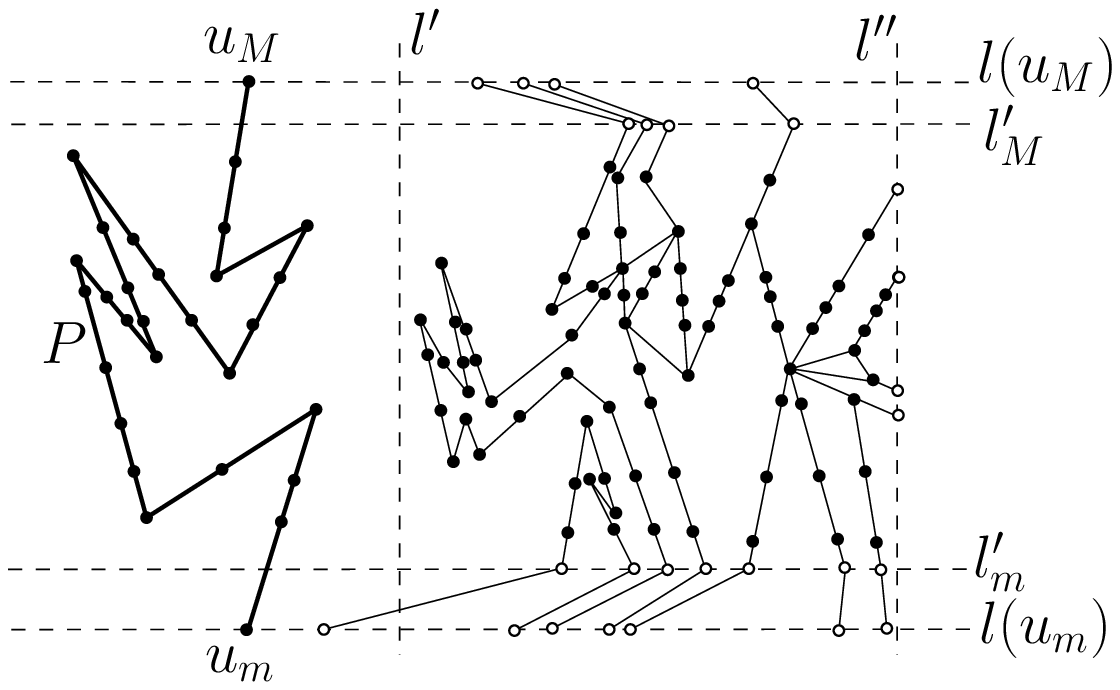}} \hspace{5mm} &
\mbox{\includegraphics[scale=0.5]{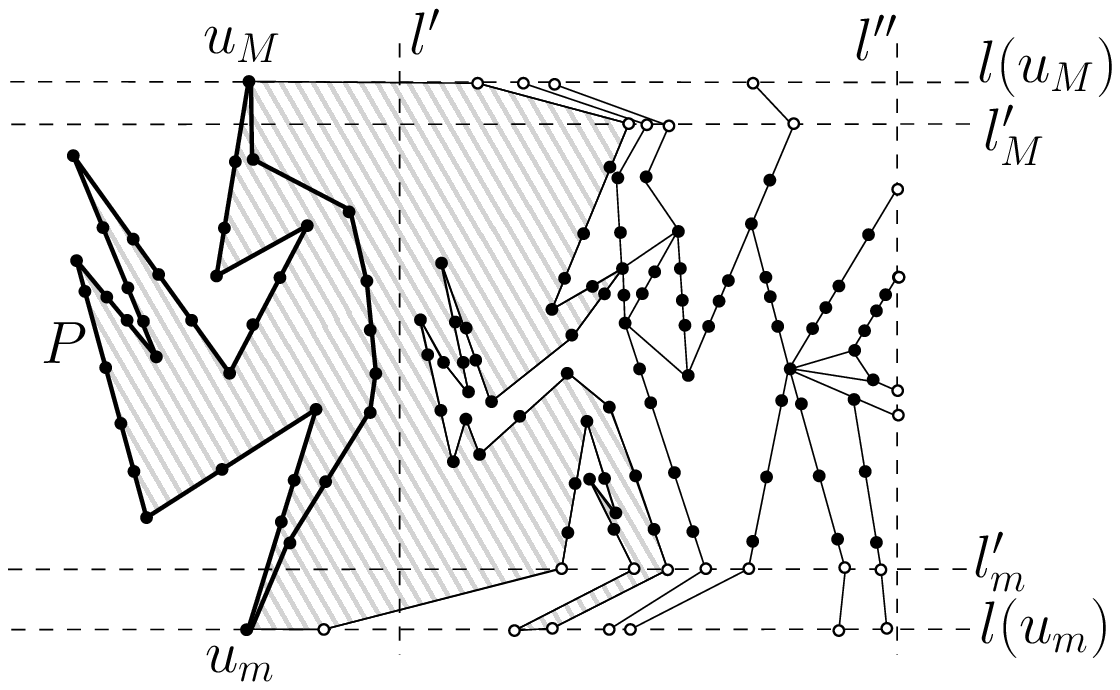}} \\
(c) \hspace{5mm}& (d)
\end{tabular}
\caption{(a) Drawing $\Gamma$ inside region $R$. The part of face $f$ inside $R$
is colored gray. Path $P$ is represented by a thick line. Intersection points of
edges with lines $l''$, $l(u_m)$, $l(u_M)$, $l'_m$, and $l'_M$ are represented
by white circles. (b) Drawing $\Gamma$ inside region $R$ after the shrinkage.
(c) Reconnecting parts of edges that have been disconnected by the shrinkage.
(d) Drawing of a monotone path connecting $u_m$ and $u_M$ inside $f$.}
\label{fig:claim3}
\end{center}
\end{figure}

We introduce some notation. Refer to Fig.~\ref{fig:claim3}(a). Let $l(u_m)$ and
$l(u_M)$ be horizontal lines through $u_m$ and $u_M$, respectively. Let $l'$ and
$l''$ be vertical lines entirely lying to the right of $P$, with $l''$ to the
right of $l'$. Denote by $D$ the distance between $l'$ and $l''$. Denote by $R$
the bounded region of the plane delimited by $P$, by $l(u_m)$, by $l(u_M)$, and
by $l''$. Denote by $y_M$ the maximum between the $y$-coordinates of the
vertices in the interior of $R$ and the $y$-coordinates of the internal vertices
of $P$. Analogously, denote by $y_m$ the minimum between the $y$-coordinates of
the vertices in the interior of $R$ and the $y$-coordinates of the internal
vertices of $P$. Denote by $y'_M$ and $y'_m$ values such that $y_M<y'_M<y(u_M)$
and $y(u_m)<y'_m<y_m$. Let $l'_m$ and $l'_M$ be the horizontal lines $y=y'_m$
and $y=y'_M$, respectively. Finally, we define some regions inside $R$. Let $R'$
be the bounded region of the plane delimited by $P$, by $l'_m$, by $l'_M$, and
by $l'$; let $R''$ be the bounded region of the plane delimited by $P$, by
$l'_m$, by $l'_M$, and by $l''$; let $R'''$ be the bounded region of the plane
delimited by $l'$, by $l'_m$, by $l'_M$, and by $l''$; let $R_B$ be the bounded
region of the plane delimited by $P$, by $l'_m$, by $l(u_m)$, and by $l''$; and
let $R_A$ be the bounded region of the plane delimited by $P$, by $l'_M$, by
$l(u_M)$, and by $l''$. We are going to modify $\Gamma$ in such a way that no
vertex and no part of an edge lies in the interior of $R'$. The part of $\Gamma$
outside $R$ is not modified in the process.

We perform an horizontal shrinkage of the part of $\Gamma$ that lies in the
interior of $R''$ (the vertices of $P$ stay still). This is done in such a way
that every intersection point of an edge with $l''$ keeps the same
$x$-coordinate, and the distance between $l''$ and every point in the part of
$\Gamma$ that used to lie inside $R''$ becomes strictly smaller than $D$. See
Fig.~\ref{fig:claim3}(b). Hence, the part of $\Gamma$ that used to lie inside
$R''$ is now entirely contained in $R'''$. However, some edges of $G$ (namely
those that used to intersect $l'_m$ and $l'_M$) are now disconnected; e.g., if
an edge of $G$ used to intersect $l'_m$, now such an edge contains a line
segment inside $R'''$, which has been shrunk, and a line segment inside $R_B$,
whose drawing has not been modified by the shrinkage. However, by construction
$R_B$ does not contain any vertex in its interior. Hence, the line segments that
lie in $R_B$ form in $\Gamma$ a planar $y$-monotone matching between a set $A$
of points on $l'_m$ and a set $B$ of points on $l(u_m)$. As a consequence of the
shrinkage, the position of the points in $A$ has been modified, however their
relative order on $l'_m$ has not been modified. Thus, we can delete the line
segments in $R_B$ and reconnect the points in $B$ with the new positions of the
points in $A$ on $l'_m$ so that each edge is $y$-monotone and no two edges
intersect. See Fig.~\ref{fig:claim3}(c). After performing an analogous
modification in $R_A$, we obtain a planar $y$-monotone drawing $\Gamma'$ of $G$
in which no vertex and no part of an edge lies in the interior of $R'$. Since no
vertex changed its $y$-coordinate and every edge is $y$-monotone, $\Gamma'$ is a
strip planar drawing of $(G,\gamma)$.

Finally, we draw a $y$-monotone curve $\cal C$ connecting $u_m$ and $u_M$. This
is done as follows. See Fig.~\ref{fig:claim3}(d). Starting from $u_m$, follow
path $P$, slightly to the right of it, until reaching line $l'_m$; continue
drawing $\cal C$ as a $y$-monotone curve in the interior of $R'$ intersecting
$l'_M$ in a point arbitrarily close to path $P$; finally, follow path $P$ until
reaching $u_M$. Place each vertex $x$ of the monotone path connecting $u_m$ and
$u_M$ on $\cal C$ at a suitable $y$-coordinate, so that $x$ lies in the strip
$\gamma(x)$. We thus obtained a strip planar drawing of $(G^+,\gamma^+)$, which
concludes the proof.
\end{proof}

Claim~\ref{cl:strict-quasijagged} implies Lemma~\ref{le:strict-quasijagged}, as
proved in the following.

First, the repetition of the above described augmentation leads to a
quasi-jagged instance $(G^*,\gamma^*)$. In fact, whenever the augmentation is performed, the number of triples
$(v_m,v_M,g)$ such that vertices $v_m$ and $v_M$ are visible local minimum
and maximum for face $g$, respectively, and such that both paths connecting
$v_m$ and $v_M$ along $C_f$ are not monotone decreases by $1$, thus eventually
the number of such triples is zero, and the instance is quasi-jagged.

Second, $(G^*,\gamma^*)$ can be constructed from $(G,\gamma)$ in polynomial
time. Namely, the number of pairs of visible local minima and maxima
for a face $g$ of $G$ is polynomial in the number of vertices of $g$. Hence, the
number of  triples $(v_m,v_M,g)$ such that vertices $v_m$ and $v_M$ are
visible local minimum and maximum for face $g$, over all faces of $G$, is
polynomial in $n$. Since a linear number of vertices
are introduced in $G$ whenever the augmentation described above is performed, it
follows that the the construction of $(G^*,\gamma^*)$ from $(G,\gamma)$ can be
accomplished in polynomial time.

Third, $(G^*,\gamma^*)$ is an instance of the strip planarity testing problem
that is equivalent to $(G,\gamma)$. This directly comes from repeated
applications of Claim~\ref{cl:strict-quasijagged}.

\subsection{From a Quasi-Jagged Instance to a Jagged
Instance}\label{subse:quasijagged-jagged}

In this section we show how to reduce a quasi-jagged instance of the strip
planarity testing problem to an equivalent jagged instance. Again, for the sake of simplicity of description, we assume that every considered instance $(G,\gamma)$ is $2$-connected.

\begin{lemma} \label{le:quasijagged-jagged}
Let $(G,\gamma)$ be a quasi-jagged instance of the strip planarity testing
problem. Then, there exists a polynomial-time algorithm that constructs an
equivalent jagged instance $(G^*,\gamma^*)$ of the strip planarity testing
problem.
\end{lemma}

Consider any face $f$ of $G$ that contains some local minimum or maximum which is not a global minimum or maximum for $f$, respectively. Assume that $f$ contains a local minimum $v$ which is not a global minimum for $f$. The case in which $f$ contains a local maximum which is not a global maximum for $f$ can be discussed analogously. Denote by $u$ (denote by $z$) the first global minimum or maximum for $f$ that is encountered when walking along $C_f$ starting at $v$ while keeping $f$ to the left (resp. to the right).

We distinguish two cases, namely the case in which $u$ is a global minimum for $f$ and $z$ is a global maximum for $f$ (Case 1), and the case in which $u$ and $z$ are both global maxima for $f$ (Case 2). The case in which $u$ is a global maximum for $f$ and $z$ is a global minimum for $f$, and the case in which $u$ and $z$ are both global minima for $f$ can be discussed symmetrically.

\begin{figure}[tb]
\begin{center}
\begin{tabular}{c c}
\mbox{\includegraphics[scale=0.5]{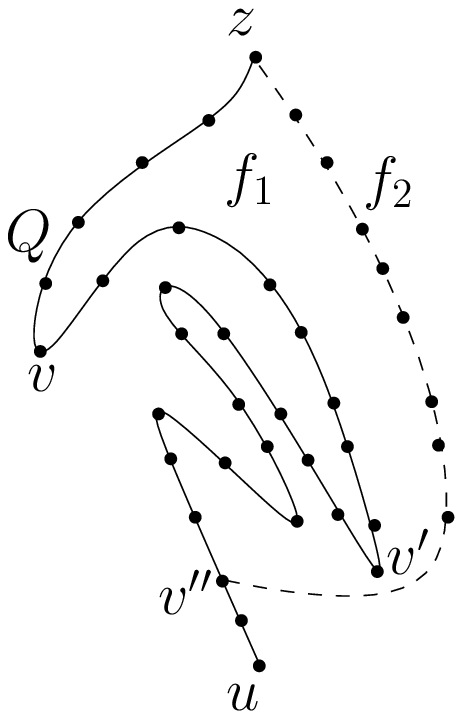}} \hspace{10mm} &
\mbox{\includegraphics[scale=0.5]{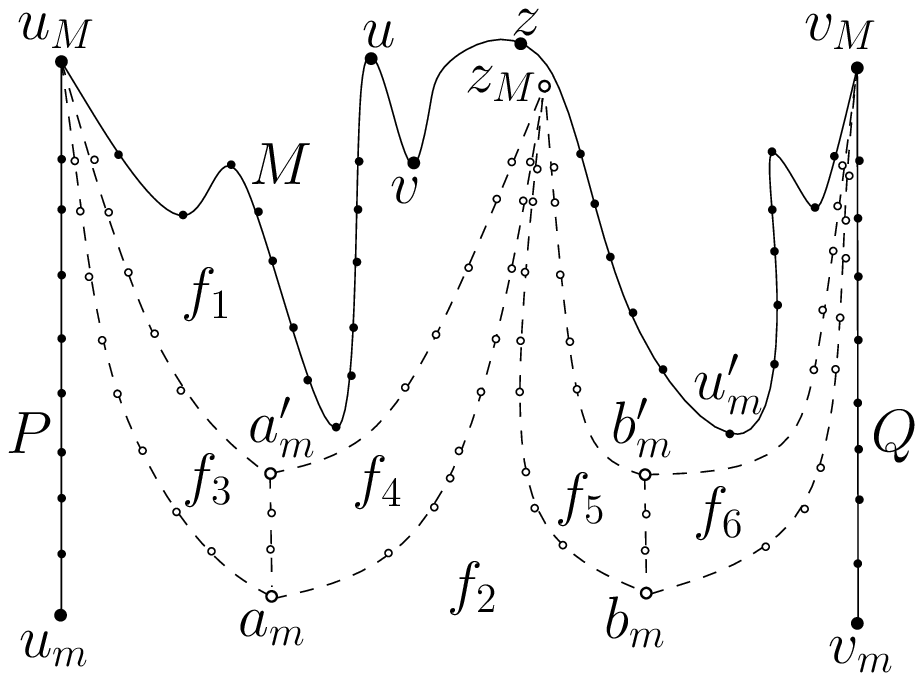}} \\
(a) \hspace{10mm}& (b)
\end{tabular}
\caption{Augmentation of $(G,\gamma)$ inside a face $f$ in: (a) Case 1 and (b) Case 2.}
\label{fig:augmentation}
\end{center}
\vspace{-5mm}
\end{figure}

In {\em Case 1}, denote by $Q$ the path connecting $u$ and $z$ in $C_f$ and containing $v$. Consider the internal vertex $v'$ of $Q$ that is a local minimum for $f$ and that is such that $\gamma(v')=\min_{u'}\gamma(u')$ among all the internal vertices $u'$ of $Q$ that are local minima for $f$. Traverse $Q$ starting from $u$, until a vertex $v''$ is found with $\gamma(v'') = \gamma(v')$. Notice that, the subpath of $Q$ between $u$ and $v''$ is monotone. Insert a monotone path connecting $v''$ and $z$ inside $f$. See Fig.~\ref{fig:augmentation}(a). Denote by $(G^+,\gamma^+)$ the resulting instance of the strip planarity testing problem. We have the following claim:

\begin{claimx} \label{cl:quasijagged-jagged1}
Suppose that Case 1 is applied to a quasi-jagged instance $(G,\gamma)$ to construct an instance $(G^+,\gamma^+)$. Then, $(G^+,\gamma^+)$ is strip planar if and only if $(G,\gamma)$ is strip planar. Also, $(G^+,\gamma^+)$ is quasi-jagged.
\end{claimx}

\remove{
\begin{proofsketch}
The necessity is trivial. For the sufficiency, consider any strip planar
drawing $\Gamma$ of $(G,\gamma)$. First, $\Gamma$ is modified so that $v''$ has
$y$-coordinate smaller than every local minimum of $Q$ different from $u$. Then,
a $y$-monotone curve can be drawn inside $f$ connecting $v''$ and $z$, thus
resulting in a strip planar drawing of $(G^+,\gamma^+)$.
\end{proofsketch}
}
\begin{proof}
We prove that $(G^+,\gamma^+)$ is strip planar if and only if $(G,\gamma)$ is
strip planar.

One direction of the equivalence is trivial, namely if $(G^+,\gamma^+)$ is strip
planar, then $(G,\gamma)$ is strip planar, since $G$ is a subgraph of $G^+$ and
$\gamma(x)=\gamma^+(x)$, for every vertex $x$ in $G$.

We prove the other direction. Consider a strip planar drawing $\Gamma$ of
$(G,\gamma)$. Observe that, since $u$ and $z$ are consecutive global minimum and
maximum for $f$, they are visible. Since $Q$ is not monotone, by assumption, and
since $(G,\gamma)$ is quasi-jagged, it follows that the path $P$ connecting $u$
and $z$ in $C_f$ and not containing $v$ is monotone. Hence, $u$ and $z$ are the
only global minimum and maximum for $f$, respectively.

For every local minimum $u'$ in $Q$ such that $\gamma(u')=\gamma(v')$ (including
$v'$), define $R(u')$ to be the bounded region delimited by the two edges
incident to $u'$ in $Q$, and by the horizontal line delimiting $\gamma(u')$ from
the top; vertically shrink $R(u')$ and the part of $\Gamma$ inside it so that
the $y$-coordinate of $u'$ is larger than the one of $v''$. Observe that such a
modification does not alter the strip planarity of $\Gamma$.

Next, we distinguish two cases.

In the first case, $f$ is an internal face of $G$. See Fig.~\ref{fig:claim4}(a).
We draw a $y$-monotone curve $\cal C$ connecting $v''$ and $z$ as follows. Draw
a line segment of $\cal C$ inside $f$ starting at $v''$ and slightly increasing
in the $y$-direction, until reaching path $P$. Then, follow such a path to reach
$z$. Place each vertex $x$ of the monotone path connecting $v''$ and $z$ on
$\cal C$ at a suitable $y$-coordinate, so that $x$ lies in the strip
$\gamma(x)$.

\begin{figure}[htb]
\begin{center}
\begin{tabular}{c c}
\mbox{\includegraphics[scale=0.5]{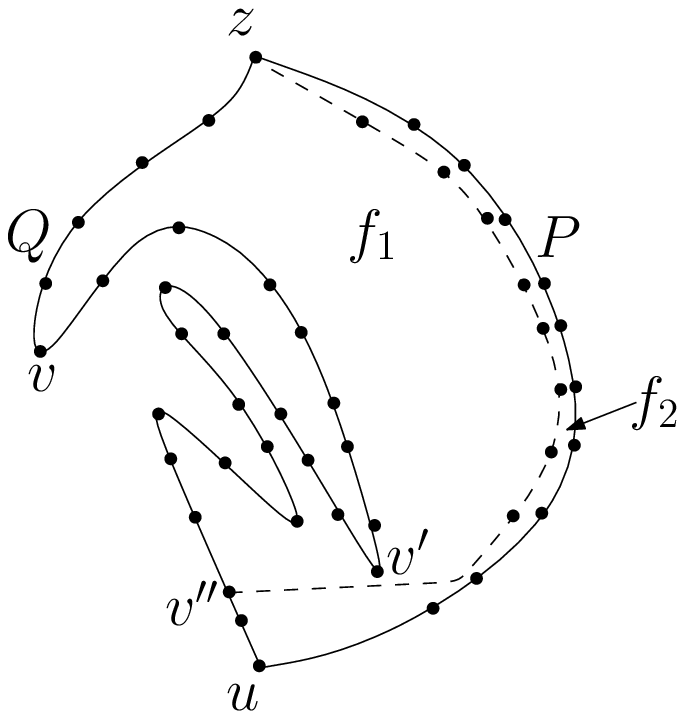}} \hspace{10mm} &
\mbox{\includegraphics[scale=0.5]{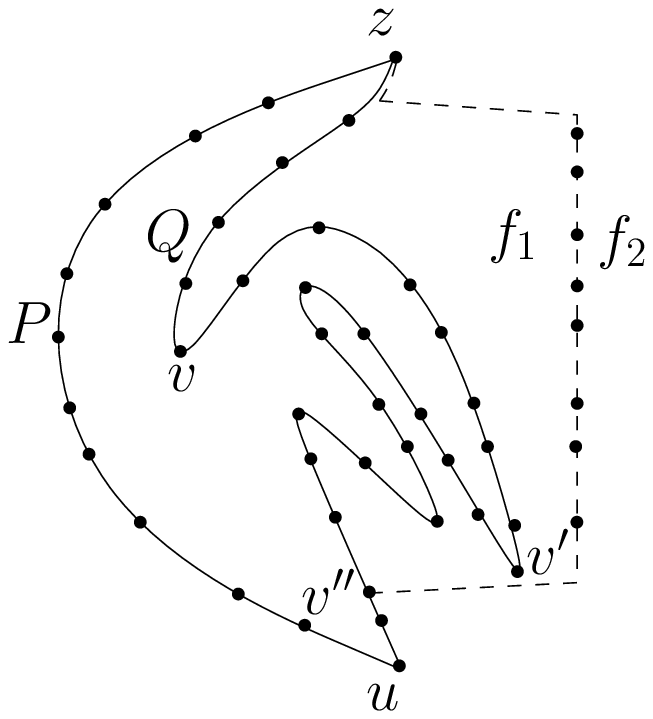}} \\
(a) \hspace{10mm}& (b)
\end{tabular}
\caption{Inserting a monotone path connecting $v''$ and $z$ inside $f$ if: (a)
$f$ is an internal face, and (b) $f$ is the outer face.}
\label{fig:claim4}
\end{center}
\end{figure}

In the second case, $f$ is the outer face of $G$. See Fig.~\ref{fig:claim4}(b).
Then, we draw a $y$-monotone curve $\cal C$ connecting $v''$ and $z$ as follows.
Draw a line segment of $\cal C$ inside $f$ starting at $v''$ and slightly
increasing in the $y$-direction, until reaching an $x$-coordinate which is
larger than the maximum $x$-coordinate of any point of $\Gamma$. Then, continue
drawing $\cal C$ as a vertical line segment, until a point is reached whose
$y$-coordinate is smaller than the $y$-coordinate of $z$ and larger than the one
of every vertex of $Q$ different from $z$ (recall that $z$ is the only global
maximum for $f$). Then, continue drawing $\cal C$  slightly increasing in the
$y$-direction and decreasing in the $x$-direction, until the edge of $Q$
incident to $z$ is reached. Then, follow such an edge to reach $z$. Place each
vertex $x$ of the monotone path connecting $v''$ and $z$ on $\cal C$ at a
suitable $y$-coordinate, so that $x$ lies in the strip $\gamma(x)$.

It remains to prove that $(G^+,\gamma^+)$ is quasi-jagged. Every face $g\neq f$
of $G$ has not been altered by the augmentation inside $f$, hence, for any two
visible local minimum $u_m$ and local maximum $u_M$ for $g$, one of the two
paths connecting $u_m$ and $u_M$ in $g$ is monotone. Denote by $f_1$ and $f_2$
the two faces into which $f$ is split by the insertion of the monotone path
connecting $v''$ and $z$, where $f_1$ is the face delimited by such a monotone
path and by the subpath of $Q$ between $v''$ and $z$. Face $f_2$ is delimited by
two monotone paths, hence the only pair of visible local minimum and local
maximum for $f_2$ is connected by a monotone path in $C_{f_2}$. Face $f_1$, on
the other hand, contains a local minimum that is not a local minimum for $f$,
namely $v''$. However, $v''$ is connected with $z$ by a monotone path in
$C_{f_1}$; also, the existence of a local maximum $u''$ for $f$ such that $v''$
and $u''$ are visible and are not connected by a monotone path in $C_{f_1}$
would imply that $u$ and $u''$ are a pair of visible local minimum and local
maximum for $f$ that is not connected by a monotone path in $C_{f}$, which
contradicts the fact that $(G,\gamma)$ is quasi-jagged.
\end{proof}

In {\em Case 2}, denote by $M$ a {\em maximal} path that is part of $C_f$, whose end-vertices are two global maxima $u_M$ and $v_M$ for $f$, that contains $v$ in its interior, and that does not contain any global minimum in its interior. By the assumptions of Case 2, such a path exists. Assume, w.l.o.g., that face $f$ is to the right of $M$ when walking along $M$ starting at $u_M$ towards $v_M$. Possibly $u_M=u$ and/or $v_M=z$. Let $u_m$ ($v_m$) be the global minimum for $f$ such that $u_m$ and $u_M$ (resp. $v_m$ and $v_M$) are consecutive global minimum and maximum for $f$. Possibly, $u_m=v_m$. Denote by $P$ the path connecting $u_m$ and $u_M$ along $C_f$ and not containing $v$. Also, denote by $Q$ the path connecting $v_m$ and $v_M$ along $C_f$ and not containing $v$. Since $M$ contains a local minimum among its internal vertices, and since $(G,\gamma)$ is quasi-jagged, it follows that $P$ and $Q$ are monotone.

Insert the plane graph $A(u_M,v_M,f)$ depicted by white circles and dashed lines in Fig.~\ref{fig:augmentation}(b) inside $f$. Consider a local minimum $u'_m \in M$ for $f$ such that $\gamma(u'_m)=\min_{v'_m}\gamma(v'_m)$ among the local minima $v'_m$ for $f$ in $M$. Set $\gamma(z_M)=\gamma(u_M)$, set $\gamma(a_m) = \gamma(b_m) = \gamma(u_m)$, and set $\gamma(a'_m) = \gamma(b'_m) = \gamma(u'_m)$. The dashed lines connecting $a_m$ and $u_M$, connecting $a'_m$ and $u_M$, connecting $a_m$ and $z_M$, connecting $a'_m$ and $z_M$, connecting $b_m$ and $z_M$, connecting $b'_m$ and $z_M$, connecting $b_m$ and $v_M$, connecting $b'_m$ and $v_M$, connecting $a_m$ and $a'_m$, and connecting $b_m$ and $b'_m$ represent monotone paths. Denote by $(G^+,\gamma^+)$ the resulting instance of the strip planarity testing problem. We have the following claim:

\begin{claimx} \label{cl:quasijagged-jagged2}
Suppose that Case 2 is applied to a quasi-jagged instance $(G,\gamma)$ to construct an instance $(G^+,\gamma^+)$. Then, $(G^+,\gamma^+)$ is strip planar if and only if $(G,\gamma)$ is strip planar. Also, $(G^+,\gamma^+)$ is quasi-jagged.
\end{claimx}

\remove{
\begin{proofsketch}
The necessity is trivial. For the sufficiency, consider any strip planar
drawing $\Gamma$ of $(G,\gamma)$. If $P$ is to the left of $Q$, then a region
$R$ is defined as the region delimited by $P$, by $M$, by $Q$, and by the
horizontal line delimiting $\gamma(u_m)$ from above. Then, the part of $\Gamma$
that lies inside $R$ is redrawn so that it lies inside a region $R_Q\subset R$
arbitrarily close to $Q$. Such a redrawing ``frees'' space for the drawing of
$A(u_M,v_M,f)$ inside $f$, which results in a strip planar drawing of
$(G^+,\gamma^+)$. If $P$ is to the right of $Q$, then $M$ might ``wiggle'' to
the right of $P$ and to the left of $Q$. Thus, we first horizontally shrink a
part of $\Gamma$ that ``intermingles'' with $P$ and $Q$, and we then draw
$A(u_M,v_M,f)$ using its four monotone paths connecting global minima with
global maxima in order to ``circumvent'' $M$. This results in a strip planar
drawing of $(G^+,\gamma^+)$.
\end{proofsketch}
}
\begin{proof}
One direction of the equivalence is trivial, namely if $(G^+,\gamma^+)$ is strip
planar, then $(G,\gamma)$ is strip planar, since $G$ is a subgraph of $G^+$ and
$\gamma(v)=\gamma^+(v)$ for every vertex $v$ in $G$.

We prove the other direction. Consider a strip planar drawing $\Gamma$ of
$(G,\gamma)$. Slightly perturb the position of the vertices in $\Gamma$ so that
no two of them have the same $y$-coordinate. Since $(G,\gamma)$ is quasi-jagged,
the path $P$ connecting $u_m$ and $u_M$ along $C_f$ and not containing $v_M$ is
monotone, and the path $Q$ connecting $v_m$ and $v_M$ along $C_f$ and not
containing $u_M$ is monotone.  We assume w.l.o.g. that face $f$ is to the right
of $P$ when traversing such a path from $u_m$ to $u_M$. Denote by $l_M$ the line
delimiting strip $\gamma(u_M)$ from below; also, denote by $l_m$ the line
delimiting strip $\gamma(u_m)$ from above.

The proof distinguishes two cases. In the first case (Case 2A), the intersection
of $P$ with $l_M$ lies to the left of the intersection of $Q$ with $l_M$. In the
second case (Case 2B), the intersection of $P$ with $l_M$ lies to the right of
the intersection of $Q$ with $l_M$. Since $P$ and $Q$ are represented in
$\Gamma$ by $y$-monotone curves that do not intersect each other, in Case 2A the
intersection of $P$ with $l_m$ lies to the left of the intersection of $Q$ with
$l_m$, while in Case 2B the intersection of $P$ with $l_m$ lies to the right of
the intersection of $Q$ with $l_m$. In both cases, we modify $\Gamma$, if
necessary, while maintaining its strip planarity so that plane graph
$A(u_M,v_M,f)$ can be planarly drawn in $f$ with $y$-monotone edges.

We first discuss Case 2A.

\begin{figure}[htb]
\begin{center}
\mbox{\includegraphics[width=0.8\textwidth]{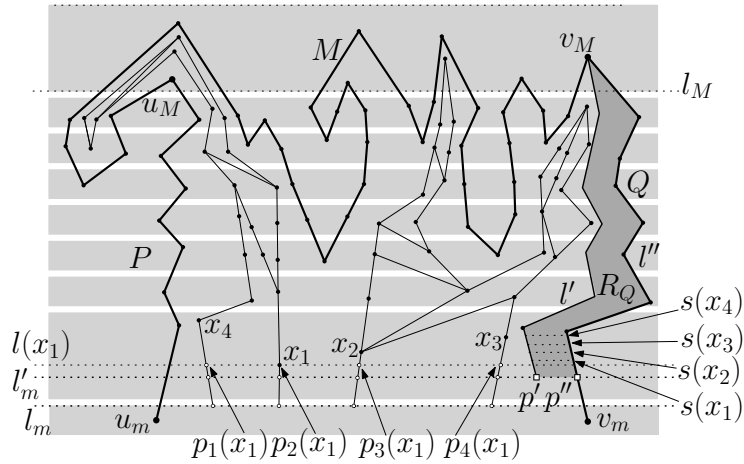}}
\caption{Illustration for the proof of Claim~\ref{cl:quasijagged-jagged2}.}
\label{fig:claim5-a}
\end{center}
\end{figure}

We introduce some notation. Refer to Fig.~\ref{fig:claim5-a}. Denote by $R$ the
bounded region of the plane delimited by $P$, by $M$, by $Q$, and by $l_m$.
Drawing $\Gamma$ will be only modified in the interior of $R$. Denote by $y_m$
the minimum between the $y$-coordinates of the vertices in the interior of $R$
and the $y$-coordinates of the internal vertices of $P$, $Q$, and $M$. Let
$y'_m$ be a value such that $y(l_m)<y'_m<y_m$. Let $l'_m$ be the horizontal line
$y=y'_m$. Denote by $R'$ the bounded region of the plane delimited by $P$, by
$M$, by $Q$, and by $l'_m$.  We define a closed bounded region $R_Q$ of the
plane inside $R$ as follows. Region $R_Q$ is delimited by two monotone curves
$l'$ and $l''$ from the left and from the right, respectively, where $l''$ is
the part of $Q$ delimited by $v_M$ and by the intersection point $p''$ of $Q$
with $l'_m$, and where $l'$ connects $v_M$ with a point $p'$ on $l'_m$, slightly
to the left of $l''$; curves $l'$ and $l''$ share no point other than $v_M$;
region $R_Q$ contains no vertex and no part of an edge of $G$ in its interior,
that is, the interior of $R_Q$ entirely belongs to $f$. Observe that a region
$R_Q$ with such properties always exists. The part of $\Gamma$ that lies in the
interior of $R'$ will be redrawn so that it entirely lies in $R_Q$.

For each vertex $x$ of $G$ that lies in the interior of $R$, consider the
horizontal line $l(x)$ through $x$. Let $p_1(x), p_2(x), \dots, p_{f(x)}(x)$ be
the left-to-right order of the intersection points of edges of $G$ with $l(x)$,
where $x$ is also a point $p_i(x)$ for some $1\leq i\leq f(x)$. We draw a
horizontal segment $s(x)$ inside $R_Q$, in such a way that: (i) $s(x)$ is
contained in the strip $\gamma(x)$, (ii) $s(x)$ connects a point in $l'$ with a
point in $l''$, and (iii) if vertices $x_1$ and $x_2$ inside $R$ are such that
$y(x_1)<y(x_2)$, then $s(x_1)$ lies below $s(x_2)$. For each vertex $x$ of $G$
that lies in the interior of $R$, insert points $p'_1(x), p'_2(x), \dots,
p'_{f(x)}(x)$ in this left-to right order on $s(x)$.

Also, let $p_1(l'_m), p_2(l'_m), \dots, p_{f(l'_m)}(l'_m)$ be the left-to-right
order of the intersection points of edges of $G$ with $l'_m$. Insert points
$p'_1(l'_m), p'_2(l'_m), \dots, p'_{f(l'_m)}(l'_m)$ in this left-to right order
on segment $\overline{p'p''}$.

We now redraw in $R_Q$ the vertices and edges that are inside $R$ in $\Gamma$.
Refer to Fig.~\ref{fig:claim5-a2}.  

\begin{figure}[htb]
\begin{center}
\mbox{\includegraphics[width=0.7\textwidth]{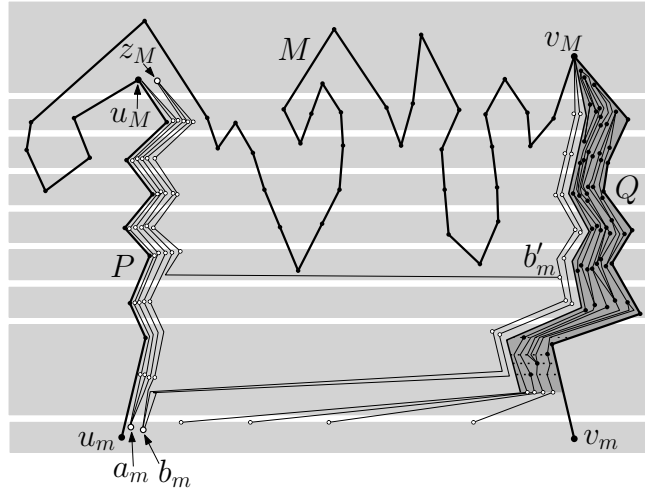}} \\
\caption{Redrawing in $R_Q$ the vertices and edges that are inside $R$ in
$\Gamma$. White circles and solid thin lines represent a drawing of
$A(u_M,v_M,f)$.}
\label{fig:claim5-a2}
\end{center}
\end{figure}

For any line segment that is part of an edge of $G$ and that connects two points
$p_i(x_1)$ and $p_j(x_2)$, with $x_1\neq x_2$, (or a point $p_i(l'_m)$ with a
point $p_j(x)$) draw a line segment connecting $p'_i(x_1)$ and $p'_j(x_2)$
(resp. connecting $p'_i(l'_m)$ with $p'_j(x)$) inside $R_Q$. Observe that, if
such a line segment exists, then $s(x_1)$ and $s(x_2)$ (resp. $\overline{p p'}$
and $s(x)$) are consecutive horizontal segments in $R_Q$. Further, the line
segments connecting points on two consecutive line segments $s(x_1)$ and
$s(x_2)$ (resp. $\overline{p p'}$ and $s(x)$) can be drawn as $y$-monotone
curves inside $R_Q$ so that they do not cross each other, give that the relative
order of the points $p'_i(x)$ on $s(x)$ preserves the order of the points
$p_i(x)$ on $l(x)$, for every vertex $x$ of $G$ in the interior of $R$, and the
relative order of the points $p'_i(l'_m)$ on $\overline{p p'}$ preserves the
order of the points $p_i(l'_m)$ on $l'_m$.

For each edge $e$ that has non-empty intersection with $R$, delete from $\Gamma$
the part $e_R$ of $e$ inside $R$. If $e$ used to intersect $l'_m$, denote by
$p_i(l_m)$ and $p_i(l'_m)$ the intersection points of $e$ with $l_m$ and $l'_m$
before $e_R$ was removed. Draw a $y$-monotone curve connecting point
$p'_i(l'_m)$ on $\overline{p p'}$ with point $p_i(l_m)$. Such curves can be
drawn without introducing crossings, given that the relative order of the points
$p'_i(l'_m)$ on $\overline{p p'}$ preserves the order of the points $p_i(l'_m)$
on $l'_m$.

We are now ready to draw $A(u_M,v_M,f)$. Draw the monotone path connecting $v_M$
with $b_m$ as a $y$-monotone curve $\cal C$ as follows. Place $b_m$ in
$\gamma(b_m)$ arbitrarily close to $P$ and to $l_m$; follow $P$ arbitrarily
close to it until reaching $l'_m$; then, continue $\cal C$ with a line segment
increasing in the $x$-direction and slightly increasing in the $y$-direction,
until reaching $l'$; then complete $\cal C$ by following $l'$ slightly to the
left of it, until reaching $v_M$. The monotone paths connecting $v_M$ with
$b'_m$ and connecting $b_m$ with $b'_m$ are arbitrarily close to the monotone
path connecting $v_M$ with $b_m$, slightly to the left of it; the $y$-coordinate
of $b'_m$ is smaller than the $y$-coordinate of every vertex of $M$. Draw the
monotone path connecting $b_m$ with $z_M$ as a $y$-monotone curve arbitrarily
close to $P$. Draw the monotone path connecting $b'_m$ with $z_M$ as a
$y$-monotone curve $\cal C'$ as follows. Start drawing $\cal C'$ from $b'_m$
with a line segment decreasing in the $x$-direction and slightly increasing in
the $y$-direction, until reaching the monotone path connecting $b_m$ and $z_M$;
then follow such a path, slightly to the right of it, until reaching $z_M$. The
remaining monotone paths lie arbitrarily close to $P$, slightly to the right of
it, and arbitrarily close to the monotone path connecting $b_m$ and $z_M$,
slightly to the left of it.

We now discuss Case 2B.

We introduce some notation. See Fig.~\ref{fig:claim5bis-1}. Denote by $l'_t$ the
horizontal line passing through the vertex $w_M$ of $M$ with largest
$y$-coordinate, and denote by $l_t$ an horizontal line in $\gamma(u_M)$ slightly
above $l_t$, and close enough to $l_t$ so that no vertex lies in the interior of
the strip delimited by $l_t$ and $l'_t$. Observe that all the vertices and edges
of $M$, of $P$, and of $Q$ are entirely below $l'_t$, except for vertex $w_M$.
Let $s(w_M)$ be a vertical segment connecting $w_M$ with $l_t$. Denote by $l'_p$
and by $l''_p$ vertical lines entirely to the right of $M$, $P$, and $Q$, with
$l''_p$ to the right of $l'_p$. Also, denote by $l'_q$ and by $l''_q$ vertical
lines entirely to the left of $M$, $P$, and $Q$, with $l''_q$ to the left of
$l'_q$. Let $R_A$ be the region delimited by $l_t$, by $l'_t$, by $l''_p$, and
by $l''_q$. Denote by $D_p$ and $D_q$ the distance between $l'_p$ and $l''_p$
and the distance between $l'_q$ and $l''_q$, respectively. Denote by $R_p$ the
bounded region of the plane delimited by $l_m$, by $l''_p$, by $l_t$, by $P$, by
the part of $M$ connecting $u_M$ with $w_M$, and by $s(w_M)$. Also, denote by
$R_q$ the bounded region of the plane delimited by $l_m$, by $l''_q$, by $l_t$,
by $Q$, by the part of $M$ connecting $v_M$ with $w_M$, and by $s(w_M)$. Drawing
$\Gamma$ will be only modified in the interior of $R_p \cup R_q$. In particular,
the vertices of $G$ and the intersection points of the edges of $G$ with the
lines delimiting $R_p \cup R_q$ will maintain the same position after the
modification.

\begin{figure}[htb]
\begin{center}
\mbox{\includegraphics[scale=0.5]{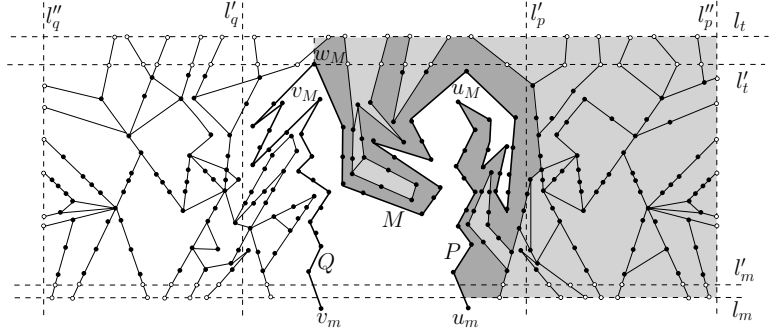}}
\caption{Drawing $\Gamma$ inside region $R_p\cup R_q$. Region $R_p$ is colored
light and dark gray. In particular, part of face $f$ inside $R_p$ is colored
dark gray. Paths $P$, $Q$, and $M$ are represented by thick lines. Intersection
points of edges with lines $l''_p$, $l''_q$, $l_m$, $l'_m$, $l_t$, and $l'_t$
are represented by white circles.}
\label{fig:claim5bis-1}
\end{center}
\end{figure}

We define some regions inside $R_p$. Let $R'_p$ be the bounded region of the
plane delimited by $l'_m$, by $l'_p$, by $l'_t$, by $P$, and by the part of $M$
connecting $u_M$ with $w_M$; let $R''_p$ be the bounded region of the plane
delimited by $l'_m$, by $l''_p$, by $l'_t$, by $P$, and by the part of $M$
connecting $u_M$ with $w_M$; let $R'''_p$ be the bounded region of the plane
delimited by $l'_m$, by $l''_p$, by $l'_p$, and by $l'_t$; finally, let
$R_{B,p}$ be the bounded region of the plane delimited by $l'_m$, by $l''_p$, by
$P$, and by $l_m$.

We analogously define some regions inside $R_q$. Let $R'_q$ be the bounded
region of the plane delimited by $l'_m$, by $l'_q$, by $l'_t$, by $Q$, and by
the part of $M$ connecting $v_M$ with $w_M$; let $R''_q$ be the bounded region
of the plane delimited by $l'_m$, by $l''_q$, by $l'_t$, by $Q$, and by the part
of $M$ connecting $v_M$ with $w_M$; let $R'''_q$ be the bounded region of the
plane delimited by $l'_m$, by $l''_q$, by $l'_q$, and by $l'_t$; finally, let
$R_{B,q}$ be the bounded region of the plane delimited by $l'_m$, by $l''_q$, by
$Q$, and by $l_m$. 

We are going to modify $\Gamma$ in such a way that no vertex and no part of an
edge lies in the interior of $R'_p \cup R'_q$. The part of $\Gamma$ outside $R_p
\cup R_q$ is not modified in the process. This modification is similar to the
one performed for the proof of Claim~\ref{cl:strict-quasijagged}. Refer to
Fig.~\ref{fig:claim5bis-2}.

We perform an horizontal shrinkage of the part of $\Gamma$ that lies inside
$R''_p$ (the vertices and edges of $P$ and $M$ stay still). This is done in such
a way that every intersection point of an edge with $l''_p$ keeps the same
$x$-coordinate, and the distance between $l''_p$ and every point in the part of
$\Gamma$ that used to lie inside $R''_p$ becomes strictly smaller than $D_p$.
Hence, the part of $\Gamma$ that used to lie inside $R''_p$ is now entirely
contained in $R'''_p$, that is the interior of $R'_p$ contains no vertex and no
part of an edge. However, some edges of $G$ (namely those that used to intersect
$l'_m$ and $l'_t$) are now disconnected; e.g., if an edge of $G$ used to
intersect $l'_m$, now such an edge contains a line segment inside $R'''_p$,
which has been shrunk, and a line segment inside $R_{B,p}$, whose drawing has
not been modified by the shrinkage. However, by construction $R_{B,p}$ does not
contain any vertex in its interior. Hence, the line segments that lie in
$R_{B,p}$ form in $\Gamma$ a planar $y$-monotone matching between a set $A_p$ of
points on $l'_m$ and a set $B_p$ of points on $l(u_m)$. As a consequence of the
shrinkage, the position of the points in $A_p$ has been modified, however their
relative order on $l'_m$ has not been modified. Thus, we can delete the line
segments in $R_{B',p}$ and reconnect the points in $B_p$ with the new positions
of the points in $A_p$ on $l'_m$ so that each edge is $y$-monotone and no two
edges intersect.

\begin{figure}[htb]
\begin{center}
\mbox{\includegraphics[scale=0.5]{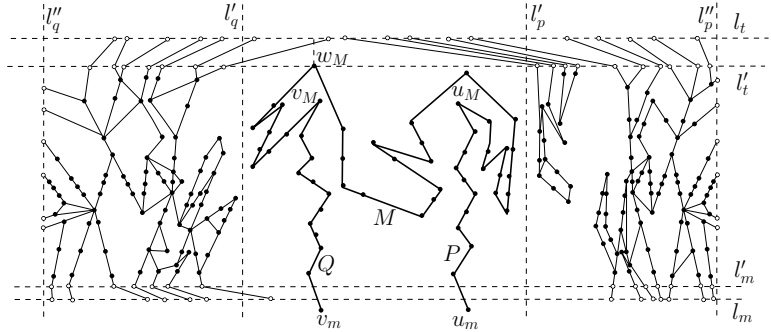}}
\caption{Drawing $\Gamma'$ of $(G,\gamma)$.}
\label{fig:claim5bis-2}
\end{center}
\end{figure}

We also perform an horizontal shrinkage of the part of $\Gamma$ that lies inside
$R''_q$ (the vertices and edges of $Q$ and $M$ stay still). This is done
symmetrically to the shrinkage of the part of $\Gamma$ that lies inside $R''_p$.
As a consequence of such a shrinkage, $R'_q$ contains no vertex and no part of
an edge.

Finally, the line segments that lie in $R_A$ form in $\Gamma$ a planar
$y$-monotone matching between a set $A'$ of points on $l'_t$ and a set $B'$ of
points on $l_t$. As a consequence of the shrinkage, the position of the points
in $A'$ has been modified, however their relative order on $l'_t$ has not been
modified. Thus, we can delete the line segments in $R_A$ and reconnect the
points in $B'$ with the new positions of the points in $A'$ on $l'_t$ so that
each edge is $y$-monotone and no two edges intersect.

We thus obtain a planar $y$-monotone drawing $\Gamma'$ of $G$ in which no vertex
and no part of an edge lies in the interior of $R'_p \cup R'_q$. Since no vertex
changed its $y$-coordinate and every edge is $y$-monotone, $\Gamma'$ is a strip
planar drawing of $(G,\gamma)$.

\begin{figure}[htb]
\begin{center}
\mbox{\includegraphics[scale=0.5]{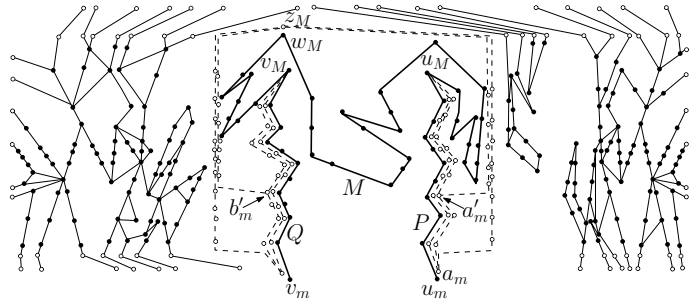}}
\caption{Drawing $A(u_M,v_M,f)$ in $\Gamma'$.}
\label{fig:claim5bis-3}
\end{center}
\end{figure}

We are now ready to draw $A(u_M,v_M,f)$. Refer to Fig.~\ref{fig:claim5bis-3}.
Place $a_m$ in point arbitrarily close to $P$, slightly to the right of it, and
slightly below $l_m$. Draw the monotone path connecting $u_M$ with $a_m$ as a
$y$-monotone curve arbitrarily close to $P$, and slightly to the right of it.
Draw the monotone path connecting $u_M$ with $a'_m$ and the monotone path
connecting $a'_m$ with $a_m$ as $y$-monotone curves arbitrarily close to the
monotone path connecting $u_M$ with $a_m$, slightly to the right of it, in such
a way that $a'_m$ has a $y$-coordinate smaller than the one of every vertex of
$P$ and $M$ in $\gamma(a'_m)$. Draw the monotone path connecting $a_m$ with
$z_M$ as a $y$-monotone curve $\cal C$ as follows. Starting from $a_m$, follow
the monotone path connecting $a_m$ with $a'_m$, slightly to the right of it,
until reaching $l'_m$. Continue drawing $\cal C$ with a line segment increasing
in the $x$-direction and slightly increasing in the $y$-direction. Just before
reaching $l'_p$, stop increasing the $x$-coordinates along $\cal C$, and
continue drawing $\cal C$ as a vertical line segment, arbitrarily close to
$l'_p$, slightly to the left of it, until reaching $l'_t$. Then, finish the
drawing of $\cal C$ with a line segment decreasing in the $x$-direction and
slightly increasing in the $y$-direction, until reaching a point on $s(w_M)$
arbitrarily close to $w_M$, on which we place $z_M$. Draw the monotone path
connecting $a'_m$ with $z_M$ as a $y$-monotone curve $\cal C'$ as follows.
Starting from $a'_m$, draw a line segment increasing in the $x$-direction and
slightly increasing in the $y$-direction, until reaching the monotone path
connecting $a_m$ with $z_M$. Then, follow such a path, slightly to the left of
it, until reaching $z_M$. Finally, the drawing of the monotone paths connecting
$v_M$ with $b_m$, connecting $v_M$ with $b'_m$, connecting $b_m$ with $b'_m$,
connecting $b_m$ with $z_M$, and connecting $b'_m$ with $z_M$ are constructed
analogously.

This concludes the construction of a strip planar drawing of $(G^+,\gamma^+)$.

It remains to prove that $(G^+,\gamma^+)$ is quasi-jagged. Every face $g\neq f$
of $G$ has not been altered by the augmentation inside $f$, hence, for any two
visible local minimum $u_m$ and local maximum $u_M$ for $g$, one of the two
paths connecting $u_m$ and $u_M$ in $G$ is monotone. Denote by
$f_1,f_2,\dots,f_6$ the faces into which $f$ is split by the insertion of
$A(u_M,v_M,f)$ (see Fig.~\ref{fig:augmentation}(b)).

Each of faces $f_3$, $f_4$, $f_5$, and $f_6$ is delimited by two monotone paths,
hence, for each $i=3,\dots,6$, the only pair of visible local minimum and local
maximum for $f_i$ is connected by a monotone path in $C_{f_i}$.

Face $f_2$ contains two local minima, namely $a_m$ and $b_m$, and one local
maximum, namely $z_M$, that are not incident to $f$. However, $u_M$ and $z_M$
are the only local maxima for $f_2$ that are visible with $a_m$; also, $a_m$ and
$b_m$ are the only local minima for $f_2$ that are visible with $z_M$; further,
$z_M$ and $v_M$ are the only local maxima for $f_2$ that are visible with $b_m$.
For all such pairs of visible local minimum and maximum, there exists a monotone
path in $C_{f_2}$ connecting them. Finally, every pair of visible local minimum
and maximum for $f_2$ which does not include $a_m$, $z_M$, or $b_m$ is also a
pair of visible local minimum and maximum for $f$, hence it is connected by the
same monotone path in $C_{f_2}$ as in $C_f$.

Analogously, each of vertices $a'_m$, $z_M$, and $b'_m$ only participates in two
pairs of visible local minimum and maximum for $f_1$, where the second vertex of
each pair is one between $u_M$, $a'_m$, $z_M$, $b'_m$, and $v_M$. For all such
pairs, monotone paths in $C_{f_1}$ exist by construction. Finally, every pair of
visible local minimum and maximum for $f_1$ which does not include $a'_m$,
$z_M$, or $b'_m$ is also a pair of visible local minimum and maximum for $f$,
hence it is connected by the same monotone path in $C_{f_1}$ as in $C_f$.
\end{proof}

Claims~\ref{cl:quasijagged-jagged1}--\ref{cl:quasijagged-jagged2} imply
Lemma~\ref{le:quasijagged-jagged}, as proved in the following.

First, we prove that the repetition of the above described augmentation leads to a jagged instance $(G^*,\gamma^*)$ of the strip planarity testing problem.
For an instance $(G,\gamma)$ and for a face $g$ of $G$, denote by $n(g)$ the number of vertices that are local minima for $g$ but not global minima for $g$,
plus the number of vertices that are local maxima for $g$ but not global maxima
for $g$. Also, let $n(G)=\sum_g n(g)$, where the sum is over all faces $g$ of $G$.
We claim that, when one of the augmentations of Cases 1 and 2 is performed
and instance $(G,\gamma)$ is transformed into an instance $(G^+,\gamma^+)$, we
have  $n(G^+)\leq n(G)-1$. The claim implies that eventually $n(G^*)=0$, hence
$(G^*,\gamma^*)$ is jagged.

We prove the claim.  When a face $f$ of $G$ is augmented as in Case~1 or in Case~2, for each face $g\neq f$ and for each vertex $u$ incident to $g$, vertex $u$ is a local minimum, a local maximum, a global minimum, or a global maximum for $g$ in $(G^+,\gamma^+)$ if and only if it is a local minimum, a local maximum, a global minimum, or a global maximum for $g$ in $(G,\gamma)$, respectively. Hence, it suffices to prove that $\sum n(f_i)\leq n(f)-1$, where the sum is over all the faces $f_i$ that are created from the augmentation inside $f$.

Suppose that Case~1 is applied to insert a monotone path between vertices $v''$ and $z$ inside $f$. Such an insertion splits $f$ into two faces, which we denote by $f_1$ and $f_2$, as in Fig.~\ref{fig:augmentation}(a). Face $f_2$ is delimited by two monotone paths, hence $n(f_2)=0$. Every vertex inserted into $f$ is neither a local maximum nor a local minimum for $f_1$. As a consequence, no vertex $x$ exists such that $x$ contributes to $n(f_1)$ and $x$ does not contribute to $n(f)$. Further, vertex $v'$ is a global minimum for $f_1$, by construction, and it is a local minimum but not a global minimum for $f$. Hence, $v'$ contributes to $n(f)$ and does not contribute to $n(f_1)$. It follows that $n(f_1)+n(f_2)\leq n(f)-1$.

Suppose that Case~2 is applied to insert plane graph $A(u_M,v_M,f)$ inside face $f$. Such an insertion splits $f$ into six faces, which are denoted by
$f_1,\dots,f_6$, as in Fig.~\ref{fig:augmentation}(b). Every vertex of $A(u_M,v_M,f)$ incident to a face $f_i$, for some $1\leq i\leq 6$, is either a
global maximum for $f_i$, or a global minimum for $f_i$, or it is neither a local
maximum nor a local minimum for $f_i$. As a consequence, no vertex $x$
exists such that $x$ contributes to some $n(f_i)$ and $x$ does not contribute to
$n(f)$. Further, for each vertex $x$ that contributes to $n(f)$, there exists at
most one face $f_i$ such that $x$ contributes to $n(f_i)$. Finally, vertex
$u'_m$ of $M$ is a global minimum for $f_1$, by construction, and it is a local
minimum but not a global minimum for $f$. Hence, $u'_m$ contributes to $n(f)$
and does not contribute to $n(f_i)$, for any $1\leq i\leq 6$. It follows that
$\sum_{i=1}^6 n(f_i) \leq n(f)-1$.

Second, $(G^*,\gamma^*)$ can be constructed from $(G,\gamma)$ in polynomial time. Namely, the number of local minima (maxima) for a face $f$ that are not global minima (maxima) for $f$ is at most the number of vertices of $f$. Hence, the number of such minima and maxima over all the faces of $G$, which is equal to $n(G)$, is linear in $n$. Since a linear number of vertices are introduced in $G$ whenever the augmentation described above is performed, and since the augmentation is performed at most $n(G)$ times, it follows that the construction of $(G^*,\gamma^*)$ can be accomplished in polynomial time.

Third, $(G^*,\gamma^*)$ is an instance of the strip planarity testing problem that is equivalent to $(G,\gamma)$. This directly comes from repeated
applications of Claims~\ref{cl:quasijagged-jagged1} and~\ref{cl:quasijagged-jagged2}.

\subsection{Testing Strip Planarity for Jagged Instances}\label{subse:jagged-upward}

In this section we show how to test in polynomial time whether a jagged
instance $(G,\gamma)$  of the strip planarity testing problem is strip planar. Recall that the associated directed graph of $(G,\gamma)$ is the directed plane graph $\overrightarrow{G}$ obtained from $(G,\gamma)$ by orienting each edge $(u,v)$ in $G$ from $u$ to $v$ if and only if
$\gamma(v)=\gamma(u)+1$. We have the following:

\begin{lemma} \label{le:upward-strip}
A jagged instance $(G,\gamma)$ of the strip planarity testing problem is strip
planar if and only if the associated directed graph $\overrightarrow{G}$ of $(G,\gamma)$ is upward
planar.
\end{lemma}

\begin{proof}
The necessity is trivial, given that a strip planar drawing of $(G,\gamma)$ is
also an upward planar drawing of $\overrightarrow{G}$, by definition.

We prove the sufficiency. A directed plane graph $\overrightarrow{G}$ is called
\emph{plane $st$-digraph} if it has exactly one source $s$ and one sink $t$ such
that $s$ and $t$ are both incident to the outer face of $\overrightarrow{G}$.
Each face $f$ of a plane $st$-digraph consists of two monotone paths called
\emph{left path} and \emph{right path}, where the left path has $f$ to the right
when traversing it from its source to its sink. 

Since $\overrightarrow{G}$ is upward planar, $\overrightarrow{G}$ can be
augmented~\cite{dt-aprad-88} to a plane $st$-digraph $\overrightarrow{G}_{st}$.
Also, this can be done by adding only \emph{dummy} edges $(u,v)$ such that $u$
and $v$ are incident to the same face $f$, and $u$ and $v$ are either both
sources or
both sinks in $C_f$ (when such a cycle is oriented according to
$\overrightarrow{G}$). Note that, since $(G,\gamma)$ is jagged, each dummy edge
$(u,v)$ is such that
$\gamma(u)=\gamma(v)$.

We now compute the \emph{directed dual} $\overrightarrow{G}_{s^* t^*}$ of
$\overrightarrow{G}_{st}$. The vertices of $\overrightarrow{G}_{s^* t^*}$ are
the faces of $\overrightarrow{G}_{st}$; two special vertices $s^*$ and
$t^*$ represent the outer face. There is an edge $(f,g)$ in
$\overrightarrow{G}_{s^* t^*}$ if face $f$ shares an edge $(u,v) \neq (s,t)$
with face $g$, and face $f$ is on the left side of $(u,v)$ when such an edge is
traversed
from $u$ to $v$. Graph $\overrightarrow{G}_{s^* t^*}$ is a plane
$st$-digraph~\cite{dt-aprad-88}.

We divide the plane into $k$ horizontal strips of fixed height, each
corresponding to one of the strips of $(G,\gamma)$.

We compute an upward planar drawing of $\overrightarrow{G}_{st}$ as follows.
First,
consider the leftmost path $p_l$ of $\overrightarrow{G}_{st}$, where
$p_l=(s=v_1^1, \dots, v_1^{h(1)}, v_2^1, \dots, v_2^{h(2)},\dots, v_k^1, \dots,
v_k^{h(k)} = t)$, with $\gamma(v_i^1) = \dots =
\gamma(v_i^{h(i)}) = i$, for $i = 1, \dots, k$. Path $p_l$ is drawn as a
$y$-monotone curve in which each vertex $u \in p_l$ lies inside strip
$\gamma(u)$.
Then, we add the faces of $\overrightarrow{G}_{st}$ one at a time, in such a way
a face is considered
after all its predecessors in $\overrightarrow{G}_{s^* t^*}$. When a face $f$ is
considered, its left path has been already drawn as a $y$-monotone curve. We
draw the right path of $f$ as a $y$-monotone curve in which each vertex $u$ lies
inside strip $\gamma(u)$. This implies that the rightmost path of the graph in
the current drawing is represented by a $y$-monotone curve.

A strip planar drawing of $(G,\gamma)$ can be obtained from the drawing of
$\overrightarrow{G}_{st}$ by removing the dummy edges.
\end{proof}

We thus obtain the following:

\begin{theorem} \label{th:algorithmic}
The strip planarity testing problem can be solved in polynomial time for instances $(G,\gamma)$ such that $G$ is a plane graph.
\end{theorem}

\begin{proof}
By Lemmata~\ref{le:general-strict}--\ref{le:quasijagged-jagged}, it is possible
to reduce in polynomial time any instance of the strip planarity testing problem
to an equivalent jagged instance $(G,\gamma)$. By Lemma~\ref{le:upward-strip}, $(G,\gamma)$ is strip planar if
and only if the associated directed plane graph $\overrightarrow{G}$ of $(G,\gamma)$ is upward
planar. Finally, by the results of Bertolazzi et al.~\cite{bdlm-udtg-94}, the
upward planarity of $\overrightarrow{G}$ can be tested in polynomial time.
\end{proof}

\section{Conclusions} \label{se:conclusions}

In this paper, we introduced the strip planarity testing problem and showed how
to solve it in polynomial time if the input graph has a prescribed plane
embedding.

We now sketch how to extend the proofs in this paper to simply-connected and even non-connected plane graphs.

Suppose that the input graph $(G,\gamma)$ is simply-connected (possibly not $2$-connected). The steps of the algorithm are the same. In particular, the transformation from a general instance to a strict instance is exactly the same. The transformation of a strict instance into a quasi-jagged instance has some differences with respect to the $2$-connected case. In fact, the visibility between local minima and maxima for a face $f$ of $G$ is now defined with respect to {\em occurrences} of such minima and maxima along $f$. Thus, the goal of such a transformation is to create an instance in which, for every face $f$ and for every pair of visible occurrences $\sigma_i(u_m)$ and $\sigma_j(u_M)$ of a local minimum $u_m$ and a local maximum $u_M$ for $f$, respectively, there is a monotone path between $\sigma_i(u_m)$ and $\sigma_j(u_M)$ in $C_f$. Such a property is achieved by using the same techniques as in Claim~\ref{cl:strict-quasijagged}. The transformation of a quasi-jagged instance into a jagged instance is almost the same as in the $2$-connected case, except that the $2$-connected components of $G$ inside a face $f$ have to be suitably squeezed along the monotone paths of $f$ to allow for a drawing of a monotone path between $v''$ and $z$ (see Case 1 of Sect.~\ref{subse:quasijagged-jagged}) or for a drawing of plane graph $A(u_M,v_M,f)$ (see Case 2 of Sect.~\ref{subse:quasijagged-jagged}). This is accomplished with the same techniques as in Claims~\ref{cl:quasijagged-jagged1} and~\ref{cl:quasijagged-jagged2}. Finally, the proof of the equivalence between the strip planarity of a jagged instance and the upward planarity of its associated directed graph does not require the instance to be $2$-connected, hence such an equivalence holds as it is.

Suppose now that the input graph $(G,\gamma)$ is not connected. Test individually the strip planarity of each connected component of $(G,\gamma)$. If one of the tests fails, then $(G,\gamma)$ is not strip planar. Otherwise, construct a strip planar drawing of each connected component of $(G,\gamma)$. Place the drawings of the connected components containing edges incident to the outer face of $G$ side by side. Repeatedly insert connected components in the internal faces of the currently drawn graph $(G',\gamma)$ as follows. If a connected component $(G_i,\gamma)$ of $(G,\gamma)$ has to be placed inside an internal face $f$  of $(G',\gamma)$, check whether $\gamma(u_M)\leq \gamma(u^f_M)$ and whether $\gamma(u_m)\geq \gamma(u^f_m)$, where $u_M$ ($u_m$) is a vertex of $(G_i,\gamma)$ such that $\gamma(u_M)$ is maximum (resp. $\gamma(u_m)$ is minimum) among the vertices of $G_i$, and where  $u^f_M$ ($u^f_m$) is a vertex of $C_f$ such that $\gamma(u^f_M)$ is maximum (resp. $\gamma(u^f_m)$ is minimum) among the vertices of $C_f$. If the test fails, then $(G,\gamma)$ is not strip planar. Otherwise, using a technique analogous to the one of Claim~\ref{cl:strict-quasijagged}, a strip planar drawing of $(G',\gamma)$ can be modified so that two consecutive global minimum and maximum for $f$ can be connected by a $y$-monotone curve $\cal C$ inside $f$. Suitably squeezing a strip planar drawing of $(G_i,\gamma)$ and placing it arbitrarily close to $\cal C$ provides a strip planar drawing of $(G'\cup G_i,\gamma)$. Repeating such an argument leads either to conclude that $(G,\gamma)$ is not strip planar, or to construct a strip planar drawing of $(G,\gamma)$.

The main question raised by this paper is whether the strip planarity testing
problem can be solved in polynomial time or is rather $\cal NP$-hard for graphs
without a prescribed plane embedding. The problem is intriguing even if the
input graph is a tree.

\remove{Finally, it would be interesting to understand whether the strip planarity
testing problem can be reduced in polynomial time to the upward planarity
testing problem and to the clustered planarity testing problem, or vice versa.}

\bibliography{Strips}
\bibliographystyle{titto-lncs-02}

\end{document}